\documentclass{article}

\usepackage{amsmath}
\usepackage{amssymb}
\usepackage{amsthm}
\usepackage{natbib}
\usepackage{hyphenat}
\usepackage{enumerate}
\usepackage{algorithm}
\usepackage{algpseudocode}
\usepackage{graphicx}
\usepackage{mathtools}
\usepackage{mathrsfs}

\newtheorem{theorem}{Theorem}

\newtheorem{lemma}[theorem]{Lemma}

\newcommand{\KLD}[2]{D_{\text{KL}}(#1 \parallel #2)}
\newcommand{\loss}[2]{L(#1 \parallel #2)}
\newcommand{\drv}[1]{\partial_{#1}}
\newcommand{\rest}[2]{%
  #1 \mathop{\mkern-3mu
  \downharpoonright
  \mkern-2mu\mathchoice{\mkern-3.5mu}{\mkern-3.5mu}{}{}_{#2}}
  \mathchoice{\mkern-1mu}{\mkern-1mu}{\mkern-.5mu}{}
}
\newcommand{\sscl}[2]{{#1}_{/#2}}
\newcommand{\subsp}[2]{\{#1, #2\}}
\newcommand{\pach}{\rightarrow}
\newcommand{\ancdec}{\rightarrow_*}
\newcommand{\rootparent}[1]{\pi_{#1}}
\newcommand{\derv}[3]{\mathcal{D}_{#1}(#2; #3)}
\newcommand{\children}[1]{k_{#1}}
\newcommand{\nss}{n_{s}}
\newcommand{\npcsp}{n_{p}}
\newcommand{\npcsprest}[1]{\rest{\npcsp}{#1}}
\newcommand{\subspcross}{\boxtimes}
\newcommand{\parentof}[2][\tau]{\pi_{#1}{(#2)}}
\newcommand{\subspsupport}[1]{\mathscr{C}_{#1}}
\newcommand{\pcspsupport}[1]{\mathscr{P}_{#1}}

\title{Variational Bayesian Supertrees}
\author{Michael Karcher, Cheng Zhang, and Frederick A Matsen IV}

\begin{document}

\maketitle

\begin{abstract}
Given overlapping subsets of a set of taxa (e.g.\ species),
and posterior distributions on phylogenetic tree topologies for each of these taxon sets,
how can we infer a posterior distribution on phylogenetic tree topologies for the entire taxon set?
Although the equivalent problem for in the non-Bayesian case has attracted substantial research,
the Bayesian case has not attracted the attention it deserves.
In this paper we develop a variational Bayes approach to this problem and demonstrate its effectiveness.
\end{abstract}

\section*{Introduction}

Fields such as phylogenetics often work with
a sort of abstracted family tree,
called a \emph{phylogenetic tree}, frequently abbreviated here as \emph{tree}.
These trees have different members of a population as their tips,
and their branching points describe the relations between the tips
and how recently they had a common ancestor.
If some of the tips are censored,
the tree topology simplifies in a process we refer to as \emph{restriction}.
If one has multiple trees restricted from the same original, uncensored tree,
one may wish to reconstruct the original \emph{supertree}.
Suppose instead one has multiple probability distributions of restricted trees,
then one may be interested in reconstructing the supertree probability distribution.
This is a difficult problem both theoretically and computationally
without additional structure.
We take a variational approach by training a flexible model to approximate
the true supertree distribution as closely as possible,
while still maintaining computational tractability.


This problem falls in the domain of \emph{supertree analysis},
a topic that has gone by this name since 1986 but has much earlier roots
as reviewed by \citet{Sanderson1998-rf} and \citet{Bininda-Emonds2004-review}.
Broadly speaking, there are two goals of supertree analysis.
The first goal is to reduce computational complexity
by dividing the ensemble of taxa into subsets,
performing independent analysis on those subsets,
and then combining these analyses into a single tree \citep{Huson1999-sh},
or in the Bayesian case a single posterior distribution.
The second goal is to combine information from multiple sources,
such as different genes,
which may have divergent phylogenetic signal and patterns of presence and absence.
Although there is some overlap between these goals,
the focus of this paper is on the first goal.
Algorithms for the second goal are better served by methods
that explicitly model the origins of different phylogenetic signal,
such as via the multispecies coalescent \citep{Liu2007-ic,Heled2010-zi}.

The eventual goal of our work is to provide a divide-and-conquer strategy
for Bayesian phylogenetics,
in which taxa are divided into subsets,
a Bayesian analysis is run on each,
and then knitted back together using a supertree approach.
Although this approach in the non-Bayesian case has been a consistent theme
in phylogenetics since the work of \citet{Huson1999-sh},
the equivalent idea in Bayesian phylogenetics is comparatively underdeveloped.
This seems surprising given that Bayesian analyses are
much more computationally demanding than their non-Bayesian counterparts,
such that the lack of rapid Bayesian inference techniques is limiting their application
in important realms such as genomic epidemiology.
In any case the lack of modern supertree methods for Bayesian analysis
is currently preventing progress on such an approach.

The most relevant existing work, by \citet{Ronquist2004-rr},
summarizes phylogenetic posterior distributions in terms of one of two schemes.
In the Weighted Independent Binary (WIB) scheme,
a tree's probability is proportional to a product of terms,
each term being present in the product of the corresponding bipartition
present in the tree.
This scheme is in a sense a simpler version of the strategy presented here.

The Weighted Additive Binary (WAB) scheme is an extension
of the long-standing tradition in supertree analysis
of performing parsimony analysis on a data matrix
formed from encoding the splits of the tree as binary characters.
The weighting in WAB comes from assigning weights to the characters
in such an encoding according to their confidences.
One can then translate the corresponding parsimony objective into a Bayesian setting
by assigning a log-likelihood penalty to each unit of parsimony cost.
In total, by taking posterior distributions for trees on each of the subsets,
summarizing them in terms of one of these schemes,
and then using products of these factors as an approximation for posterior probability.
\citet{Ronquist2004-rr} show some correlation of this method
with actual tree posterior probabilities for example data sets on six and ten taxa,
and that the WAB scheme outperforms the WIB scheme.

In this paper we develop a variational Bayesian formulation of supertree inference.
Given a set of reference distributions 
of tree topologies with overlapping tips,
we find a \emph{supertree distribution} on the entire tip set
that closely approximates each reference distribution
when only considering the tips in that reference.
We will model these reference distributions and our supertree distribution
using subsplit Bayesian networks (\emph{SBNs}) \citep{sbn} reviewed below,
which generalize previous formalisms for describing probability distributions
on topologies~\citep{Hhna2012-pm,Larget2013-uo}.
We note in passing that these formalisms, in turn, noted connections between
their methods and the supertree work of \citet{Ronquist2004-rr}.
We focus on the case where the reference distributions are
originally given as, or subsequently approximated by, SBNs,
but the method is generalizable to arbitrary reference distributions
at the cost of computational efficiency.
We accomplish our goal of training a supertree distribution
using gradient descent to minimize the differences between
our reference distributions and our supertree distribution
(appropriately restricted).
Moreover, we show that the method successfully trains
a supertree distribution that is close to the truth
on both simulated and real-world phylogenetic sequence data.

\section*{Methods}

\subsection*{Overview}

Suppose we are given a set of probability distributions $\{p_i\}$
on rooted, bifurcating phylogenetic tree topologies,
abbreviated as \emph{tree topologies} or simply \emph{topologies},
each with a corresponding tip set $X_i$.
We refer to these \emph{tree distributions} on their respective tip sets
as our \emph{reference distributions}.
Our method for reconstructing a supertree distribution
is to find a probability distribution $q(\tau)$ of topologies
on the entire taxon set $X = \cup_i X_i$ so that $q$ is as close as possible to each of
the reference distributions when the tips in not present in that reference distribution are removed.
Figure~\ref{fig:supertree_merge} illustrates the flow of information,
taking two reference distributions and producing a supertree distribution
on the union of its references' tip sets.

\begin{figure}
  \centering
  \includegraphics[width=0.7\textwidth]{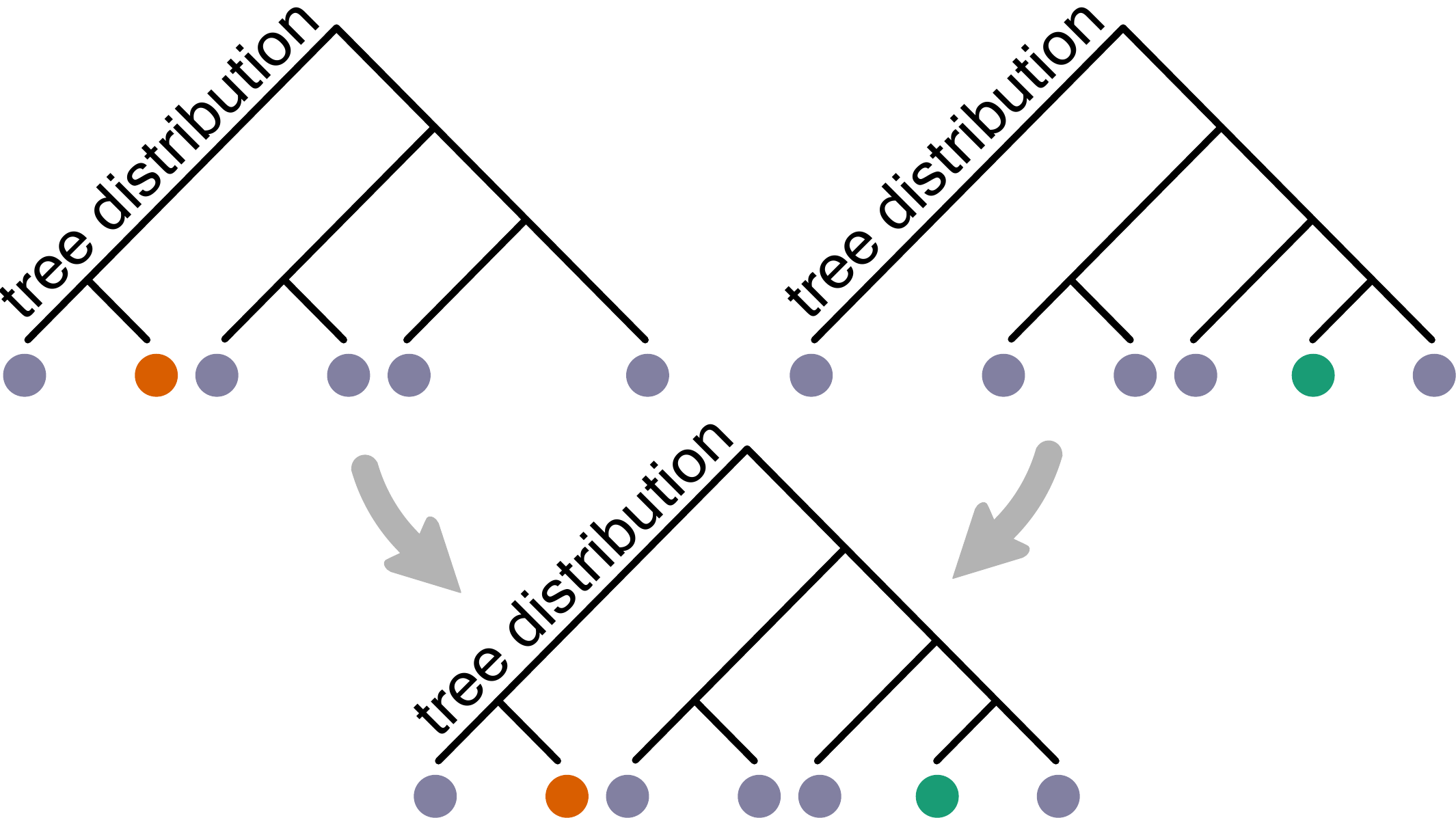}
  \caption{Illustrating the supertree method on distributions.
  The top row of two tree distributions are reference distributions,
  and the bottom tree distribution is the supertree distribution.
  The gray circles indicate tips shared by the two reference distributions,
  whereas the orange and green circles represent tips
  that are only present in one of the two references.}
  \label{fig:supertree_merge}
\end{figure}

\begin{figure}
  \centering
  \includegraphics[width=0.6\textwidth]{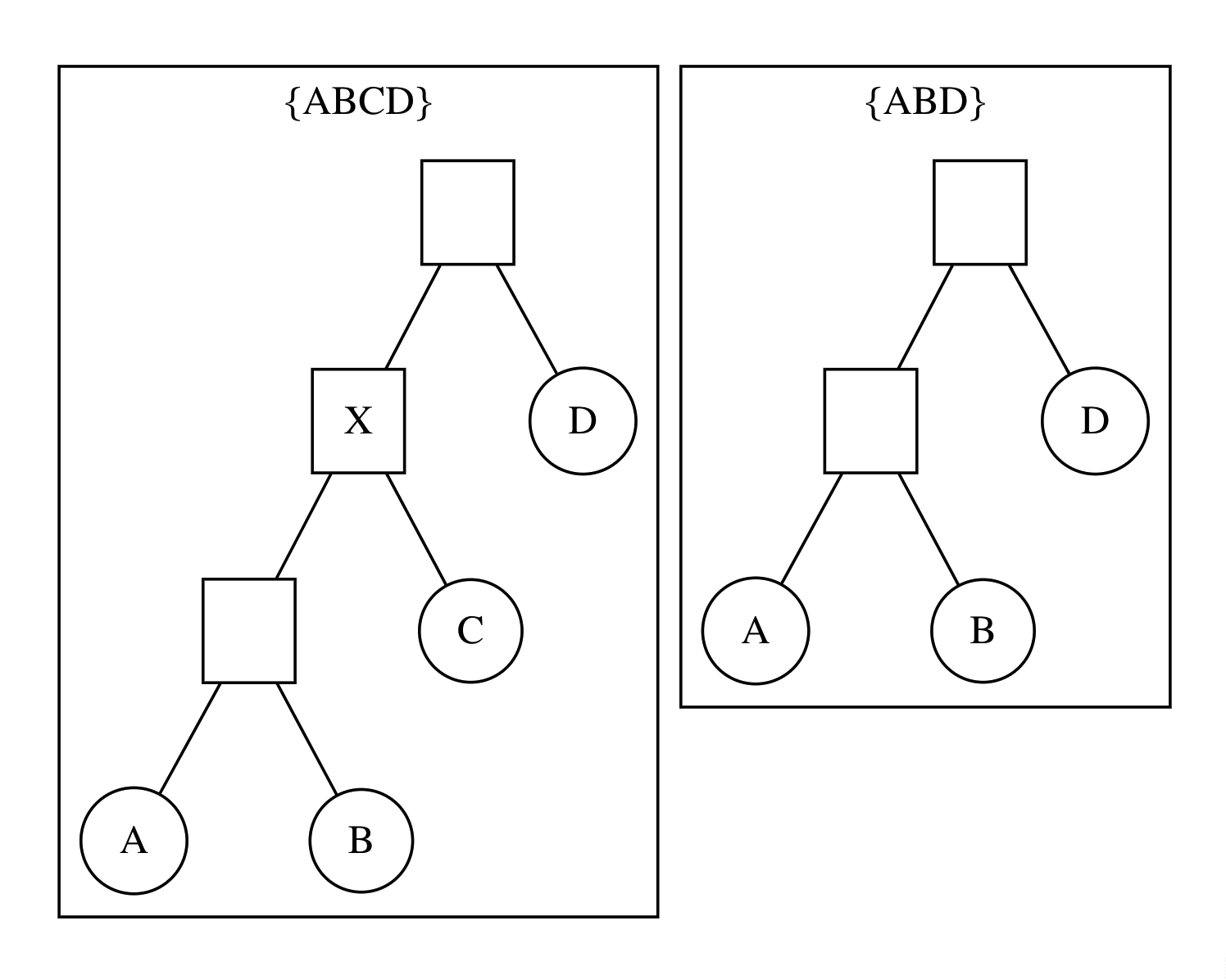}
  \caption{Illustration of restricting tree topologies.
  We restrict a tree with tip set $\{ABCD\}$ (left)
  to tip set $\{ABD\}$ (right).
  We remove tip $C$ and the internal node marked with an $X$
  as it no longer has two descendants.
  }%
  \label{fig:tree_restriction_example}
\end{figure}

We now establish the formalisms necessary to achieve this goal.
Given a taxon subset $\bar X \subset X$ and a topology $\tau$ on $X$,
define the \emph{restriction} $\rest{\tau}{\bar{X}}$ to be the topology induced by
removing all taxa that are not in $\bar X$ from $\tau$
and removing all internal nodes with fewer than two children \citep{Semple2003-em}.
We illustrate an example of this process in
Figure~\ref{fig:tree_restriction_example}.
Given a probability distribution $q$ on tree topologies with taxon set $X$
and a topology $\bar{\tau}$ on $\bar X$,
we define the restriction of $q$ to $\bar X$ as a marginalization over
the topologies that restrict to $\bar{\tau}$,
\begin{equation}
\rest{q}{\bar{X}}(\bar \tau) \coloneqq \sum_{\tau: \rest{\tau}{\bar{X}} = \bar{\tau}} q(\tau).
\label{eq:marginalSum}
\end{equation}

In this paper, our goal is to infer a distribution $q(\tau)$ for topologies on the entire taxon set $X$
such that its restrictions $\{\rest{q}{X_i}\}$ are close to
a set of distributions $\{p_i\}$ for topologies on corresponding taxon subsets $\{X_i\}$.
Specifically, our main objective will be to minimize
our \emph{loss function}: the sum of KL-divergences between
each given distribution and $q$ restricted to its taxon subset,
\[
\loss{\{p_i\}}{q} = \sum_i \KLD{p_i}{\rest{q}{X_i}}.
\]
If we have reason to prioritize some reference distributions differently than others,
due to differing confidence in the different distributions among other reasons,
we can easily incorporate weights into a weighted loss function,
\[
\loss{\{p_i, w_i\}}{q} = \sum_i w_i \KLD{p_i}{\rest{q}{X_i}}.
\]
Hereafter, though, we focus on the unweighted version.

For parameterizations of $q$ such that
$\KLD{p_i}{\rest{q}{X_i}}$ has an efficiently computable gradient
with respect to $q$'s parameters,
gradient descent is available for minimizing the loss function.
We describe one such family of parameterizations using SBNs below
and derive efficient KL-divergences and gradients later in this section.

\subsection*{Review of subsplit Bayesian networks}

Here we review subsplit Bayesian networks (SBNs) in the case of rooted topologies.
Our approach will have a different emphasis than the original \citet{sbn} work---where
the previous work described SBNs as very general class of Bayesian networks
and concentrated on unrooted trees,
we will focus on a simpler SBN structure parameterizing
rooted (phylogenetic) trees.

We will use the term \emph{clade} to refer to a subset of a taxon set $X$.
A \emph{subsplit} is a set containing two disjoint \emph{child clades} $s = \subsp{Y}{Z}$.
We define the \emph{parent clade} of a subsplit as the union of its two child clades,
with notation $U(s) = Y \cup Z$.
If we need to specify a particular child clade $Y$ of a subsplit $s$
as being the focus of attention (as opposed to the other child clade),
we use the notation $\sscl{s}{Y}$.
Note that the parent clade of $s$ is allowed to be a subset of $X$,
in contrast to the traditional definition of a \emph{split} as a bipartition of the entire taxon set $X$~\citep{Semple2003-em}.
We will say that the subsplit $s$ \emph{divides} a clade $W$ if $U(s) = W$,
with notation $W \pach s$.
We also say that $s$ is a \emph{child subsplit} of \emph{parent subsplit} $t$
if $U(s)$ is a child clade of $t$.
We refer to $t$ and $s$ as a \emph{parent-child subsplit pair} or \emph{PCSP}
with notation $t \pach s$ when $s$ is a valid child of a subsplit $t$.

We also extend the concept of valid child subsplits to further descendants.
Given subsplit $a = \subsp{Y}{Z}$, we say that subsplit $d$ is a valid \emph{descendant} of $\sscl{a}{Y}$ if $U(d) \subseteq Y$,
and we use the notation $\sscl{a}{Y} \ancdec d$.
Additionally, we say that $d$ is a valid descendant of $a$
with notation $a \ancdec d$ if
\begin{align*}
  &\sscl{a}{Y} \ancdec d,\\
  &\sscl{a}{Z} \ancdec d, \text{or}\\
  &a = d.
\end{align*}
Equivalently, we say $a$ is a valid \emph{ancestor} of $d$ under the same conditions
and with the same notation.
Note that $t \pach s$ implies $t \ancdec s$ but not vice versa.

We use the term \emph{path} to refer to a sequence of subsplits
such that each element is a descendant of the previous.
For example, the path $a \ancdec t \pach s$ would refer to
a sequence starting with $a$,
proceeding via any number of subsplits to $t$ (including zero if $a=t$),
then directly to $t$'s child $s$.

It will be convenient to also introduce \emph{singletons} and \emph{trivial subsplits}.
A singleton corresponds to one of the tips of the tree
and is represented by a clade with size one
or a subsplit containing a singleton clade and the empty set.
A subsplit is trivial if one of its child clades is empty.
We typically exclude singletons and trivial subsplits from sets of subsplits,
unless explicitly included.

Each bifurcating rooted topology can be uniquely represented
as a set of the subsplits it contains.
For example, the topology given by the Newick string \citep{newickformat}
 \texttt{"((t1,t2),t3);"} is described by
the subsplits $\subsp{\{\mathtt{t1},\mathtt{t2}\}}{\{\mathtt{t3}\}}$
and $\subsp{\{\mathtt{t1}\}}{\{\mathtt{t2}\}}$.
We will use the notation $s \in \tau$ to mean that the subsplit $s$ is found in $\tau$
and the notation $\tau \subseteq S$ to mean that all of the subsplits in $\tau$ are in set $S$.
The same holds true for specifying a topology in terms of PCSPs,
and we will use the same notation in that case.
Each subsplit $s \in \tau$ has two child clades
which each must correspond to a singleton or a subsplit that divides it.
Similarly, for a given topology $\tau$ each tip or subsplit $s$ has a
\emph{parent subsplit} $t$ such that $U(s)$ is one of the child clades of $t$.
We will denote the parent subsplit of $s$ with $\parentof{s}$.
In the above example,
$\subsp{\{\mathtt{t1},\mathtt{t2}\}}{\{\mathtt{t3}\}}$
is the parent of $\subsp{\{\mathtt{t1}\}}{\{\mathtt{t2}\}}$,
which in turn is the parent of $\{\mathtt{t1}\}$ and $\{\mathtt{t2}\}$.
In order to eliminate having to make a special case for the root subsplit $r$,
we define its parent subsplit to be a special trivial subsplit of the entire taxon set,
i.e. $\parentof{r} = \subsp{X}{\emptyset} = \rootparent{X}$.

In order to illustrate how to construct an SBN,
we first describe how to sample a topology from an SBN\@.
Starting from the root clade,
recursively construct a topology: for
any currently-childless clade $W$ larger than a singleton,
sample a subsplit that divides $W$ from
a probability distribution, supplied by the SBN,
conditional on some subset of the ancestors of $W$.
These conditional distributions can be parameterized in different ways
using different subsets of the clades' ancestry \citep{sbn},
with each parameterization defining a family of SBN probability distributions
on tree topologies.
In this paper, we focus on two families in particular:
\emph{clade-conditional distributions} (CCDs) where the subsplit distributions
$p(s|U(s))$ are conditional on the subsplits' parent clade~\citep{Hhna2012-pm,Larget2013-uo},
and \emph{subsplit-conditional distributions} (SCDs) where the subsplit distributions
$p(s|\sscl{t}{U(s)})$ are conditional on the subsplits' parent subsplit and clade~\citep{sbn}.
We fix the conditional probability of any singleton to be 1.
With our induced conditional independence assumptions,
the SBN probability for a rooted tree $\tau$ can then be easily computed:
under CCDs $p(\tau) = \prod_{s \in \tau} p(s | U(s))$,
and under SCDs $p(\tau) = \prod_{s \in \tau} p(s | \sscl{\parentof{s}}{U(s)})$.

We use the notation $p(s) \coloneqq p(s \in \tau)$
for the \emph{unconditional} probability
of a subsplit $s$ being present in a topology $\tau$ randomly sampled according to $p$.
Similarly, we use $p(t \pach s)$ for the unconditional probability of PCSP $t \pach s$.
For CCD-parameterized SBN $p$,
we define the \emph{subsplit support} $\subspsupport{}$ as the set of subsplits that
have positive probability under $p$.
For SCD-parameterized SBN $p$, we define the \emph{PCSP support} $\pcspsupport{}$ as a heterogeneous set containing
the PCSPs that have positive probability under $p$,
the subsplits that have positive probability under $p$,
the singletons for $p$'s tip set,
and the empty subsplit.

\subsection*{KL-divergence between SBNs}

Here we show that the KL-divergence
between two SBN-parameterized distributions can be computed efficiently.
If both $p(\tau)$ and $q(\tau)$ are CCD-parameterized SBNs,
\begin{align*}
  \KLD{p}{q} &= -\sum_{\tau} p(\tau) \log \left( \frac{q(\tau)}{p(\tau)} \right) \\
                 &= -\sum_{\tau} p(\tau) \sum_{s} 1_{s \in \tau} \log \left( \frac{q(s|U(s))}{p(s|U(s))} \right) \\
                 &= -\sum_{s} \log \left( \frac{q(s|U(s))}{p(s|U(s))} \right) \sum_{\tau} p(\tau) 1_{s \in \tau} \\
                 &= -\sum_{s} p(s) \, \left[ \log \left( q(s|U(s)) \right) - \log \left( p(s|U(s)) \right) \right]. \stepcounter{equation}\tag{\theequation}\label{eq:KL_CCD}
\end{align*}
Computing this sum is linear time in the number of subsplits in the subsplit support of $p$.

Similarly, if both $p(\tau)$ and $q(\tau)$ are SCD-parameterized SBNs,
\begin{align*}
  \KLD{p}{q} &= -\sum_{\tau} p(\tau) \log \left( \frac{q(\tau)}{p(\tau)} \right) \\
    &= -\sum_{\tau} p(\tau)
    \sum_{(t \pach s)} 1_{(t \pach s) \in \tau}
    \log \left( \frac{q(s|\sscl{t}{U(s)})}{p(s|\sscl{t}{U(s)})} \right) \\
    &= -\sum_{(t \pach s)} \log \left( \frac{q(s|\sscl{t}{U(s)})}{p(s|\sscl{t}{U(s)})} \right)
    \sum_{\tau} p(\tau) 1_{(t \pach s) \in \tau} \\
    &= -\sum_{(t \pach s)} p(t \pach s) \:
    \left[ \log \left( q(s|\sscl{t}{U(s)}) \right)
    - \log \left( p(s|\sscl{t}{U(s)}) \right) \right]. \stepcounter{equation}\tag{\theequation}\label{eq:KL_SCD}
\end{align*}
Computing this sum is linear time in the number of PCSPs
in the PCSP support of $p$.

\subsection*{Restricting SBNs}

Equation~\ref{eq:marginalSum} defines how to take
a distribution $q$ on trees with taxon set $X$
and restrict it to its induced distribution on
trees with taxon set $\bar{X} \subset X$.
If $q$ is an SBN-parameterized distribution,
we can more efficiently calculate $\rest{q}{\bar{X}}$ from the SBN parameters directly:
we can restrict a subsplit $s$ to taxon set $\bar{X}$
by taking the intersection of both child clades with $\bar{X}$.
One consequence of this is that some subsplits on $X$
will become trivial subsplits on $\bar{X}$.
On the other hand, if a restricted subsplit is nontrivial,
then we know the original subsplit is nontrivial,
because the restricted subsplit separates at least one pair of tips,
so the original subsplit will separate those tips at well.
Furthermore, subsplits represent recursive bipartitions of sets,
so any pair of tips can only be partitioned by a subsplit once.
Therefore,
all subsplits that restrict to the same nontrivial subsplit $\bar{s}$ are mutually exclusive,
since any subsplit that restricts to $\bar{s}$
separates all the same tips that $\bar{s}$ partitions.
By this mutual exclusivity, the probability of a restricted subsplit $\bar{s}$
under a restricted distribution $\rest{q}{\bar{X}}$ is simply
\begin{equation} \label{eq:CCDUnconditional}
  \rest{q}{\bar{X}}(\bar{s}) = \sum_{s: \, \rest{s}{\bar{X}} = \bar{s}} q(s).
\end{equation}
Similarly, subsplits with the same clade are mutually exclusive,
so the unconditional probability of a clade appearing is
\begin{equation} \label{eq:CCDCladeUnconditional}
  \rest{q}{\bar{X}}(\bar{U}) = \sum_{\bar{s}': \, U(\bar{s}')=\bar{U}} \rest{q}{\bar{X}}(\bar{s}').
\end{equation}

In order to construct the restricted SBN,
we need to compute the appropriate conditional probabilities,
which we can easily calculate from unconditional probabilities.
In a CCD context, subsplit $s$ probabilities are conditional on observing its clade $U(s)$.
We can build upon Equation~\ref{eq:marginalSum} to find the restricted SBN induced by restricting to $\bar{X} \subset X$.
We see
\begin{equation} \label{eq:CCDConditional}
  \rest{q}{\bar{X}}(\bar{s} \mid  U(\bar{s})) = \frac{\rest{q}{\bar{X}}(\bar{s}, U(\bar{s}))}{\rest{q}{\bar{X}}(U(\bar{s}))} = \frac{\rest{q}{\bar{X}}(\bar{s})}{\rest{q}{\bar{X}}(U(\bar{s}))}.
\end{equation}

\begin{figure}
  \centering
  \includegraphics[width=0.7\textwidth]{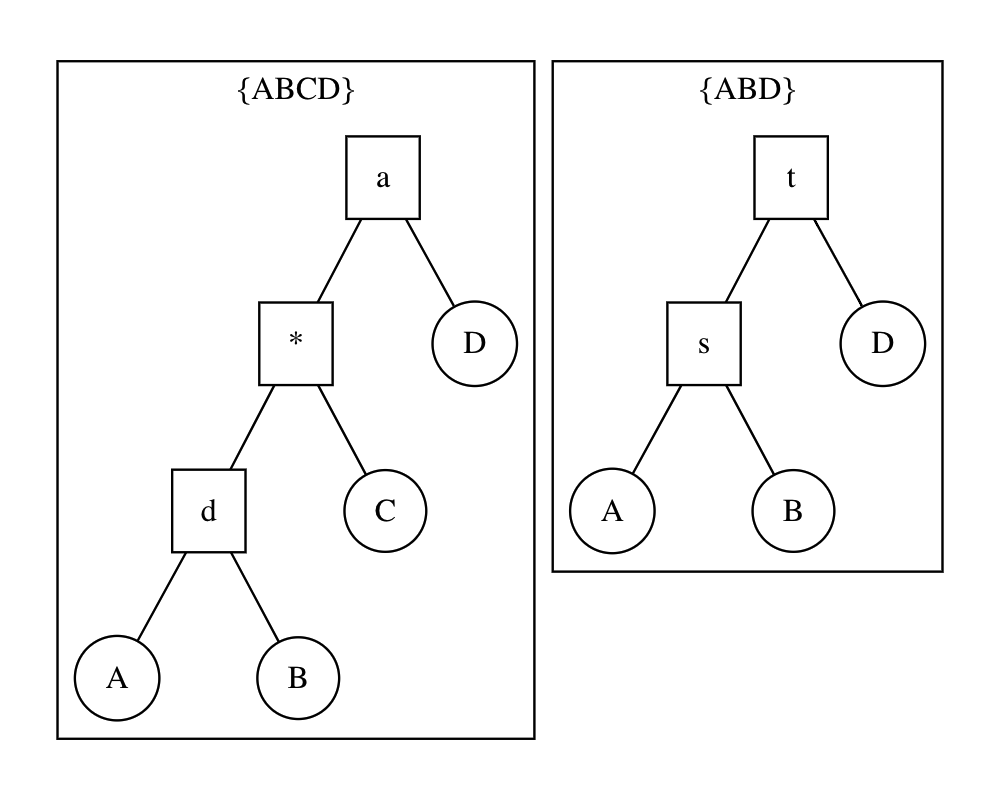}
  \caption{Illustration of restricting PCSPs.
  We restrict a tree with tip set $\{ABCD\}$ (left)
  to tip set $\{ABD\}$ (right).
  Note how the node with label $a$ on the left
  corresponds to the node with label $t$ on the right,
  just as the node with label $d$ corresponds to the node with label $s$.
  Note that nodes $a$ and $d$ are not parent and child,
  but nodes $t$ and $s$ are.
  This is possible because the node with label $*$ is trivial under restriction,
  and therefore has no corresponding node on the right.
  }%
  \label{fig:restriction_example}
\end{figure}

A slightly more involved construction is needed to gain an equivalent formula for the SCD case.
For SCD parameterizations, we need the unconditional probability of a PCSP\@.
PCSPs are mutually exclusive, so an argument similar to the above holds,
but subsplits $a$ and $d$ that respectively restrict to the restricted parent $\bar{t}$ and child $\bar{s}$
do not themselves have to be a valid parent-child pair before restriction.
We illustrate this possibility in Figure~\ref{fig:restriction_example}.
More formally, the following two statements are equivalent:
(1) the ancestor-descendant pair $a \ancdec d$
restricts to the PCSP $\bar{t} \pach \bar{s}$,
and (2) a sequence of parent-child pairs in $q$ exists,
starting with $a$, ending with $d$, with subsplits $\{t_i\}$ in between,
such that each $\rest{t_i}{\bar{X}}$ is trivial.
The converse is elementary,
but to show that (1) implies (2),
we know $a$ and $d$ exist by assumption,
and $a$ restricts to $\bar{t}$ which is the parent of $\bar{s}$.
Then one of $a$'s child clades restricts to $U(\bar{s})$,
and every subsplit between $a$ and $d$ must divide $U(\bar{s})$ under restriction.
Finally, since the tips in $U(\bar{s})$ can only be partitioned once (in $\bar{s}$),
every subsplit between $a$ and $d$ must be trivial under restriction.

We use the notation $q(a \ancdec d)$ to represent the probability
of observing subsplits $a$ and $d$ in a random tree from $q$
if $d$ is a valid descendant of $a$ and zero otherwise.
The unconditional probability of a PCSP under a restricted distribution is then,
\begin{equation*}
  \rest{q}{\bar{X}}(\bar{t} \pach \bar{s}) = 1_{\{{\bar{t}} \pach {\bar{s}}\}} \sum_{a:\, \rest{a}{\bar{X}} = \bar{t}} \, \sum_{d:\, \rest{d}{\bar{X}} = \bar{s}} q(a \ancdec d).
\end{equation*}
Then under restriction, the conditional probability of a PCSP given its parent is
\begin{equation} \label{eq:SCDConditional}
  \rest{q}{\bar{X}}(\bar{s} | \bar{t}) = \frac{\rest{q}{\bar{X}}({\bar{t}} \pach {\bar{s}})} {\rest{q}{\bar{X}}(\bar{t})}.
\end{equation}

\subsection*{Supertree support}
Our overall goal is to find a distribution $q$ on topologies that is close
to a set of reference distributions $\{p_i\}$ on taxon sets $X_i$.
An important part of that goal is to understand the \emph{supertree support},
namely the set of building blocks (subsplits or PCSPs)
that have positive probability under $q$,
and is isomorphic to the set of trees that have positive probability under $q$.
We assume that we are given supports for each of our reference distributions $p_i$,
which we will call our \emph{reference supports}.
We will refer to process of finding a mutual supertree support
for the entire taxon set as \emph{mutualizing} or \emph{mutualization}.
The details of how this is done will depend on
whether we are using a CCD or SCD parameterization.

We seek to find a suitable supertree support
for the sake of computational tractability,
so we wish to have as few elements in our supertree support as reasonably possible.
However, any tree that restricts to a tree in each reference support
is as suitable for inclusion in the supertree support as any other,
so we must attempt to include them all.
We codify these objectives as a pair of Requirements,
stated here generally and later more specifically in CCD and SCD contexts.
\begin{enumerate}[{Requirement} 1:]
  \item\label{itm:gen_mut_req_1} To allow an element into the supertree support,
  it must restrict to elements in each reference support,
  \item\label{itm:gen_mut_req_2} Any tree that, under restriction,
  is a subset of every reference support, must be included in the supertree support.
\end{enumerate}

For more than two reference supports,
we propose an incremental approach for building the mutual support:
we start with taxon set $X_1$ and its reference support,
extend to $X_1 \cup X_2$ by mutualizing with the reference support for $X_2$,
then continue to $(X_1 \cup X_2) \cup X_3$, etc.
Thus we will only present an algorithm for the case of $X = X_1 \cup X_2$.
The algorithms extend to finding supertree supports
for multiple sets simultaneously,
but the computations grow exponentially in the number of simultaneous supports
(see Discussion).

\subsubsection*{CCD subsplit supports}

Assume that we have reference subsplit supports $\subspsupport{X_i}$ for each taxon subset $X_i \subset X$
and wish to find a good candidate subsplit support $M(\{\subspsupport{X_i}\})$ for the supertree distribution $q(\tau)$.
We now specialize the Requirements to the CCD case:
\begin{enumerate}[{CCD Requirement} 1:]
  \item\label{itm:mut_req_1} If $s \in M(\{\subspsupport{X_i}\})$, then for each $i$, $\rest{s}{X_i} \in \subspsupport{X_i}$,
  \item\label{itm:mut_req_2} For every tree $\tau$ on $X$ such that, for each $i$, $\rest{\tau}{X_i} \subseteq \subspsupport{X_i}$ then $\tau \subseteq M(\{\subspsupport{X_i}\})$.
\end{enumerate}
Requirement~\ref{itm:mut_req_1} says that any subsplit in the mutualized support must exist (after restriction) in each of the input subsplit supports.
Requirement~\ref{itm:mut_req_2} says that any topology that appears (after restriction) in each of the restricted supports must be present in the mutualized support.
These are in fact fairly strong constraints.
For example, if reference supports do not agree on overlapping sets of taxa,
then the supertree support can be too small or even empty.
However, if the reference supports are restrictions of the true supertree support,
or are supersets of the true support restrictions,
then the mutualized support will cover the true support.
Below we present an algorithm that fulfills these requirements.

Next we explain the Requirements in more detail for the case of two reference supports.
For any $W \subseteq X$,
define $\subspsupport{X}(W)$ as
all subsplits in $\subspsupport{X}$ that divide $W$,
including the trivial subsplit $\subsp{W}{\emptyset}$.
In order for a subsplit $s$ to meet Requirement~\ref{itm:mut_req_1},
$s$ must be a member of both $\subspsupport{X_1}(W)$
and $\subspsupport{X_2}(W)$ after restriction.

To find a collection of subsplits that satisfies Requirement~\ref{itm:mut_req_2},
we take an iterative approach over possible clades
from the root to the tips (Algorithm~\ref{alg:mutual_subsplit}).
Starting with $W=X$,
we consider every pairing of subsplits
$s_1,s_2 = \subsp{Y_1}{Z_1}, \subsp{Y_2}{Z_2}$ from $\subspsupport{X_1}(W) \times \subspsupport{X_2}(W)$.
For each pair,
we generate a set of two potential subsplits
$s_1 \subspcross s_2 \coloneqq
\left\{\subsp{Y_1 \cup Y_2}{Z_1 \cup Z_2},
\subsp{Y_1 \cup Z_2}{Z_1 \cup Y_2} \right\}$.
Note that potential subsplits will frequently have taxa in both child clades,
but we will exclude these invalid subsplits.
We add each nontrivial valid subsplit $s \in s_1 \subspcross s_2$ to the output,
and we add each child clade of $s$ to the stack of clades to consider
if the clade is size two or larger and it has not been visited before.

\begin{algorithm}
\caption{Mutual Subsplit Support algorithm $M(\subspsupport{X_1}, \subspsupport{X_2})$}
\begin{algorithmic}[1]
\State{} Input $\subspsupport{X_1}, \subspsupport{X_2}$
\State{} $X \coloneqq X_1 \cup X_2$
\State{} Stack $[X]$, Visited $\{\}$, Output $\{\}$
\While{Stack not empty}
  \State{} $W \coloneqq$ $\text{pop}(\text{Stack})$
  \State{} add $W$ to Visited
  \State{} $W_1 \coloneqq W \cap X_1, W_2 \coloneqq W \cap X_2$
  \ForAll{$s_1, s_2 \in \subspsupport{X_1}(W_1) \times \subspsupport{X_2}(W_2)$}
  \ForAll{valid, nontrivial $s = \subsp{Y}{Z} \in s_1 \subspcross s_2$}
      \State{} add $s$ to Output
      \If{$Y$ not in Visited and $|Y| \geq 2$}
        push $Y$ to Stack
      \EndIf{}
      \If{$Z$ not in Visited and $|Z| \geq 2$}
        push $Z$ to Stack
      \EndIf{}
    \EndFor{}
  \EndFor{}
\EndWhile{}
\State{} return Output
\end{algorithmic}%
\label{alg:mutual_subsplit}
\end{algorithm}

Here we prove that Algorithm~\ref{alg:mutual_subsplit} meets both of our requirements for the supertree subsplit support.
This first theorem establishes the first Requirement in the CCD case.

\begin{theorem}
  $\rest{M(\subspsupport{X_1}, \subspsupport{X_2})}{X_1} \subseteq \subspsupport{X_1}$
  and $\rest{M(\subspsupport{X_1}, \subspsupport{X_2})}{X_2} \subseteq \subspsupport{X_2}$.
\end{theorem}
\begin{proof}
  By construction, every subsplit $s \in M(\subspsupport{X_1}, \subspsupport{X_2})$
  $X_1$-restricts to a subsplit $\rest{s}{X_1} \in \subspsupport{X_1}$
  and $X_2$-restricts to a subsplit $\rest{s}{X_2} \in \subspsupport{X_2}$.
\end{proof}

\begin{lemma}\label{lem:subsplit_mut}
  Suppose there exists a tree topology $\tau$ with taxon set $X$
  such that $\rest{\tau}{X_1} \subseteq \subspsupport{X_1}$
  and $\rest{\tau}{X_2} \subseteq \subspsupport{X_2}$.
  If $s \in \tau$ and the algorithm reaches state $W = U(s)$,
  then $s \in M(\subspsupport{X_1}, \subspsupport{X_2})$.
\end{lemma}
\begin{proof}
  Suppose subsplit $s = \subsp{Y}{Z}$,
  then we know
  $\rest{s}{X_1} = \subsp{Y \cap X_1}{Z \cap X_1}$
  and $\rest{s}{X_2} = \subsp{Y \cap X_2}{Z \cap X_2}$.
  By assumption $W = U(s)$,
  $\rest{s}{X_1} \in \subspsupport{X_1}$, and $\rest{s}{X_2} \in \subspsupport{X_2}$,
  so the algorithm considers
  subsplits in $\subspsupport{X_1}$ that divide $W_1 \coloneqq W \cap X_1$,
  and subsplits in $\subspsupport{X_2}$ that divide $W_2 \coloneqq W \cap X_2$.
  We know
  \begin{align*}
    W_1 &= W \cap X_1 = U(s) \cap X_1 = (Y \cup Z) \cap X_1 \\
        &= (Y \cap X_1) \cup (Z \cap X_1) = U(\rest{s}{X_1}).
  \end{align*}
  Similarly, $W_2 = U(\rest{s}{X_2})$,
  so the algorithm considers $\rest{s}{X_1}$ from $\subspsupport{X_1}$
  and $\rest{s}{X_2}$ from $\subspsupport{X_2}$ at this step.
  Then one of the subsplits on $X$ that the algorithm generates is
  \begin{align*}
    \subsp{[Y \cap X_1] \cup [Y \cap X_2]}{[Z \cap X_1] \cup [Z \cap X_2]} &= \subsp{Y \cap [X_1 \cup X_2]}{Z \cap [X_1 \cup X_2]} \\
    &= \subsp{Y \cap X}{Z \cap X} = s.
  \end{align*}
  We know $s$ is valid and nontrivial,
  so $s \in M(\subspsupport{X_1}, \subspsupport{X_2})$.
\end{proof}

We can now establish the second Requirement in the CCD case.
\begin{theorem}
  If there exists a tree topology $\tau$ with taxon set $X$
  such that $\rest{\tau}{X_1} \subseteq \subspsupport{X_1}$
  and $\rest{\tau}{X_2} \subseteq \subspsupport{X_2}$,
  then $\tau \subseteq M(\subspsupport{X_1}, \subspsupport{X_2})$.
\end{theorem}
\begin{proof}
  We use a proof by induction recursively
  from root to tips over the subsplits in $\tau$.
  Our base case is the root split of $\tau$.
  The algorithm begins with $W = X$ in the stack,
  so by Lemma~\ref{lem:subsplit_mut} $\tau$'s root split will be in the output.
  Our inductive step for general $s \in \tau$
  allows us to assume that $s$'s parent $\parentof{s}$
  has already been visited and is already in the output.
  Since $\parentof{s} \pach s$ and $s \in \tau$,
  we know $U(s) \in \parentof{s}$ and $|U(s)| > 1$,
  so $U(s)$ will already be in the stack
  and the algorithm will eventually reach state $W = U(s)$.
  Then by Lemma~\ref{lem:subsplit_mut} $s$ will be in the output.
  Therefore, by induction, all $s \in \tau$ are in $M(\subspsupport{X_1}, \subspsupport{X_2})$,
  so $\tau \subseteq M(\subspsupport{X_1}, \subspsupport{X_2})$.
\end{proof}

\subsubsection*{SCD PCSP supports}

For SCD-parameterized SBNs, we need to consider the supertree PCSP support.
Suppose that we have reference PCSP supports $\pcspsupport{X_i}$ for each taxon subset $X_i \subset X$.
We have requirements for PCSP support mutualization that parallel
our requirements for constructing a subsplit support.
Our Requirements take the following form in the SCD case:
\begin{enumerate}[{SCD Requirement} 1:]
  \item\label{itm:pcsp_mut_req_1} If $({t}\pach{s}) \in M(\{\pcspsupport{X_i}\})$,
  then for each $i$,
  there exists a path $(a_i \ancdec t \pach s) \subset M(\{\pcspsupport{X_i}\})$
  and a subsplit $u_i$ in $\pcspsupport{X_i}$
  such that $\rest{a_i}{X_i} = u_i$ and $U(\rest{s}{X_i}) \in u_i$.
  \item\label{itm:pcsp_mut_req_2} For every tree $\tau$ on $X$ such that,
  for each $i$, $\rest{\tau}{X_i} \subseteq \pcspsupport{X_i}$
  then $\tau \subseteq M(\{\pcspsupport{X_i}\})$.
\end{enumerate}
\noindent Please see Figure~\ref{fig:restriction_example} for an
illustration of Requirement 1.

Our PCSP mutualization algorithm, laid out in Algorithm~\ref{alg:mutual_pcsp},
largely follows the structure of Algorithm~\ref{alg:mutual_subsplit},
with a few subtle differences.
We use the notation $\pcspsupport{X}(\sscl{t}{W})$ to represent
the set of all valid child subsplits of $\sscl{t}{W}$ in $\pcspsupport{X}$,
including the trivial subsplit $\subsp{W}{\emptyset}$.
Because parent subsplits can become trivial under restriction,
we need to perform additional bookkeeping in Algorithm~\ref{alg:mutual_pcsp}.
The items in our recursion stack (line~\ref{alg:line:pcsp_stack})
contain three pieces of information:
the parent subsplit and clade under consideration in the full taxon set,
the most recent parent subsplit in $\pcspsupport{X_1}$,
and the most recent parent subsplit in $\pcspsupport{X_2}$.
The additional \textbf{if} statements (lines~\ref{alg:line:pcsp_if1} and~\ref{alg:line:pcsp_if2})
in the main loop capture both most recent parent subsplits in their respective supports.

\begin{algorithm}
\caption{Mutual PCSP Support algorithm $M(\pcspsupport{X_1}, \pcspsupport{X_2})$}
\begin{algorithmic}[1]
\State{} Input $\pcspsupport{X_1}, \pcspsupport{X_2}$
\State{} $X \coloneqq X_1 \cup X_2$
\State{} Stack $[(\sscl{\rootparent{X}}{X}, \sscl{\rootparent{X_1}}{X_1}, \sscl{\rootparent{X_2}}{X_2})]$, Visited $\{\}$, Output $\{\}$
\While{Stack not empty}
  \State{} $(\sscl{t}{W}, \sscl{t_1}{W_1}, \sscl{t_2}{W_2}) \coloneqq \text{pop}(\text{Stack})$\label{alg:line:pcsp_stack}
  \State{} add $(\sscl{t}{W}, \sscl{t_1}{W_1}, \sscl{t_2}{W_2})$ to Visited
  \ForAll{$s_1, s_2 \in \pcspsupport{X_1}(\sscl{t_1}{W_1}) \times \pcspsupport{X_2}(\sscl{t_2}{W_2})$}
    \ForAll{valid $s = \subsp{Y}{Z} \in s_1 \subspcross s_2$}
      \State{} add $(t \pach s)$ to Output
      \State{} $Y_1 \coloneqq Y \cap X_1, Y_2 \coloneqq Y \cap X_2$
      \State{} $Z_1 \coloneqq Z \cap X_1, Z_2 \coloneqq Z \cap X_2$
      \If{$s_1 \in \pcspsupport{X_1}$}\label{alg:line:pcsp_if1}
        \State{} $u_1 \coloneqq s_1$
      \Else{}
        \State{} $u_1 \coloneqq t_1$
      \EndIf{}
      \If{$s_2 \in \pcspsupport{X_2}$}\label{alg:line:pcsp_if2}
        \State{} $u_2 \coloneqq s_2$
      \Else{}
        \State{} $u_2 \coloneqq t_2$
      \EndIf{}
      \If{$|Y| \geq 2$ and $(\sscl{s}{Y}, \sscl{u_1}{Y_1}, \sscl{u_2}{Y_2})$ not in Visited}
        \State{} push $(\sscl{s}{Y}, \sscl{u_1}{Y_1}, \sscl{u_2}{Y_2})$ to Stack
      \EndIf{}
      \If{$|Z| \geq 2$ and $(\sscl{s}{Z}, \sscl{u_1}{Z_1}, \sscl{u_2}{Z_2})$ not in Visited}
        \State{} push $(\sscl{s}{Z}, \sscl{u_1}{Z_1}, \sscl{u_2}{Z_2})$ to Stack
      \EndIf{}
    \EndFor{}
  \EndFor{}
\EndWhile{}
\State{} return Output
\end{algorithmic}%
\label{alg:mutual_pcsp}
\end{algorithm}

The need for and operation of these additional constructions is best illustrated with an example, which we provide in the Appendix.
This example uses only 4 taxa and is written in great detail.

Next we prove that Algorithm~\ref{alg:mutual_pcsp} meets both of our requirements for the PCSP subsplit support.
This first theorem establishes the first Requirement in the SCD case.
\begin{theorem}
  For every $(t \pach s) \in M(\pcspsupport{X_1}, \pcspsupport{X_2})$ and each $i = 1,2$,
  there exists a path $(a_i \ancdec t \pach s) \subseteq M(\pcspsupport{X_1}, \pcspsupport{X_2})$
  and a subsplit $u_i$ in $\pcspsupport{X_i}$
  such that $\rest{a_i}{X_i} = u_i$ and $U(\rest{s}{X_i}) \in u_i$.
\end{theorem}
\begin{proof}
  For any parent subsplit $t$,
  every candidate child subsplit $s$ is constructed from the children
  (possibly trivial)
  of the most recent ancestors present in their respective reference support $\{\rest{a_i}{X_i} = u_i\}$.
  By induction, we know $a_i \ancdec t$.
  Finally, by construction $t \pach s$, $\rest{a_i}{X_i} = u_i$,
  and $U(\rest{s}{X_i}) \in u_i$.
\end{proof}

We now begin preparations for the proof of the second Requirement in the SCD case.
\begin{lemma}\label{lem:pcsp_restriction}
  If $\rest{\tau}{\bar{X}} \subseteq \pcspsupport{\bar{X}}$,
  then for every PCSP $(t \pach s) \in \tau$,
  there exists a path $(a \ancdec t \pach s) \subseteq \tau$
  and a subsplit $u$ in $\pcspsupport{\bar{X}}$
  such that $\rest{a}{\bar{X}} = u$ and $U(\rest{s}{\bar{X}}) \in u$.
  If $\rest{s}{\bar{X}}$ is nontrivial,
  then $(u \pach \rest{s}{\bar{X}}) \in \pcspsupport{\bar{X}}$.
\end{lemma}
\begin{proof}
  If $\rest{t}{\bar{X}}$ is a subsplit in $\pcspsupport{\bar{X}}$,
  we know $U(\rest{s}{\bar{X}}) \in \rest{t}{\bar{X}}$,
  so $u = \rest{t}{\bar{X}}$, and we are done.
  If $\rest{t}{\bar{X}}$ is not a subsplit in $\pcspsupport{\bar{X}}$,
  i.e. $\rest{t}{\bar{X}}$ is trivial, but is not $\rootparent{\bar{X}}$ nor a singleton,
  then we know $U(\rest{s}{\bar{X}}) = U(\rest{t}{\bar{X}}) \in \rest{\parentof{t}}{\bar{X}}$.
  If $\rest{\parentof{t}}{\bar{X}}$ is a subsplit in $\pcspsupport{\bar{X}}$
  then we are done.
  Otherwise, we can continue this reasoning and chain of equalities,
  proceeding up the tree until we reach a parent subsplit in $\pcspsupport{\bar{X}}$ or $\rootparent{\bar{X}}$,
  which we know is a parent subsplit in $\pcspsupport{\bar{X}}$.
  If $\rest{s}{\bar{X}}$ is nontrivial,
  then there is an uninterrupted path of trivial subsplits
  between $s$ and the $a$ we find above.
  Thus, by tree restriction,
  $(u \pach \rest{s}{\bar{X}})$ is a valid PCSP in $\pcspsupport{\bar{X}}$.
\end{proof}

\begin{lemma}\label{lem:pcsp_mut}
  Suppose there exists a tree topology $\tau$ with taxon set $X$
  such that $\rest{\tau}{X_1} \subseteq \pcspsupport{X_1}$
  and $\rest{\tau}{X_2} \subseteq \pcspsupport{X_2}$.
  If $(t \pach s) \in \tau$ and the algorithm reaches
  state $(\sscl{t}{U(s)}, \sscl{\rest{a_1}{X_1}}{U(s) \cap X_1}, \sscl{\rest{a_2}{X_2}}{U(s) \cap X_2})$,
  where $a_1$ is the most recent ancestor of $s$
  that restricts to a subsplit in $\pcspsupport{X_1}$,
  and $a_2$ is most recent ancestor of $s$
  that restricts to a subsplit in $\pcspsupport{X_2}$,
  then $(t \pach s) \in M(\pcspsupport{X_1}, \pcspsupport{X_2})$.
\end{lemma}
\begin{proof}
  Similar to Algorithm~\ref{alg:mutual_subsplit},
  in general we know
  \begin{align*}
    W_1 &\coloneqq U(s) \cap X_1 = (Y \cup Z) \cap X_1 \\
        &= (Y \cap X_1) \cup (Z \cap X_1) = U(\rest{s}{X_1}),
  \end{align*}
  and via the same logic $W_2 = U(\rest{s}{X_2})$.
  Let $u_1 = \rest{a_1}{X_1}$ and $u_2 = \rest{a_2}{X_2}$.
  By Lemma~\ref{lem:pcsp_restriction},
  we know $W_1 \in u_1$ and $W_2 \in u_2$.
  If $s_1 = \rest{s}{X_1}$ is trivial,
  then $U(s_1) = W_1$ and is in $\pcspsupport{X_1}(\sscl{t_1}{W_1})$ by construction.
  If $s_1$ is nontrivial,
  then also by Lemma~\ref{lem:pcsp_restriction}
  we know $s_1 \in \pcspsupport{X_1}(\sscl{t_1}{W_1})$.
  Similarly, $s_2 = \rest{s}{X_2} \in \pcspsupport{X_1}(\sscl{t_2}{W_2})$,
  so Algorithm~\ref{alg:mutual_pcsp} considers $s_1$ and $s_2$ at this step.
  Then one of the subsplits on $X$ that the algorithm generates is
  \begin{align*}
    \subsp{[Y \cap X_1] \cup [Y \cap X_2]}{[Z \cap X_1] \cup [Z \cap X_2]} &=
    \subsp{Y \cap [X_1 \cup X_2]}{Z \cap [X_1 \cup X_2]} \\
    &= \subsp{Y \cap X}{Z \cap X} = s.
  \end{align*}
  We know $s$ is nontrivial,
  so $(t \pach s)$ is added to $M(\pcspsupport{X_1}, \pcspsupport{X_2})$.
\end{proof}

The following theorem establishes the second Requirement in the SCD case.
\begin{theorem}
  If there exists a tree topology $\tau$ with taxon set $X$
  such that $\rest{\tau}{X_1} \subseteq \pcspsupport{X_1}$
  and $\rest{\tau}{X_2} \subseteq \pcspsupport{X_2}$,
  then $\tau \subseteq M(\pcspsupport{X_1}, \pcspsupport{X_2})$.
\end{theorem}
\begin{proof}
  We proceed via recursive proof by induction over all PCSPs $(t \pach s) \in \tau$,
  with $s = \subsp{Y}{Z}$.
  Our base case is the algorithm's first state
  $\sscl{t}{W} = \sscl{\rootparent{X}}{X}$,
  $t_1 = \rootparent{X_1}$, and $t_2 = \rootparent{X_2}$.
  We know $s$ only has one ancestor, $\rootparent{X}$,
  which restricts to $\rootparent{X_1}$ in $\pcspsupport{X_1}$
  and to $\rootparent{X_2}$ in $\pcspsupport{X_2}$.
  This satisfies the criteria of Lemma~\ref{lem:pcsp_mut},
  so $(t \pach s)$ will be in $M(\pcspsupport{X_1}, \pcspsupport{X_2})$.

  Our inductive step for general $(t \pach s) \in \tau$ uses the same argument,
  but uses the inductive assumption
  that PCSP $(\parentof{t} \pach t)$ was previously added to the output.
  Given this assumption,
  we will show that the algorithm constructs the state triplet
  $(\sscl{t}{U(s)}, \sscl{\rest{a_1}{X_1}}{U(s) \cap X_1}, \sscl{\rest{a_2}{X_2}}{U(s) \cap X_2})$.
  If the algorithm reaches this state,
  then Lemma~\ref{lem:pcsp_mut} guarantees that
  $(t \pach s)$ will be in $M(\pcspsupport{X_1}, \pcspsupport{X_2})$.
  Lemma~\ref{lem:pcsp_restriction} guarantees that such
  $\rest{a_1}{X_1}$ and $\rest{a_2}{X_2}$ exist.

  Since $(\parentof{t} \pach t)$ has already been visited and $|U(s)| > 2$,
  we know that there is a triplet in the stack
  with $\sscl{t}{U(s)}$ as the first component.
  Considering the other components of this triplet,
  if $t$ restricts to a subsplit in $\pcspsupport{X_1}$,
  then $a_1 = t$ and $u_1 = \rest{t}{X_1}$.
  If not, then $t$ and $s$ have the same most recent ancestor
  that restricts to a subsplit in $\pcspsupport{X_1}$,
  so the algorithm passes it along as $u_1$.
  Either way, $\sscl{\rest{a_1}{X_1}}{U(s) \cap X_1}$ is the second component.
  The same argument holds for $X_2$ and the third component.
  This is exactly the state triplet we require,
  so by Lemma~\ref{lem:pcsp_mut},
  we know $(t \pach s) \in M(\pcspsupport{X_1}, \pcspsupport{X_2})$.
  Therefore, by induction,
  all $(t \pach s) \in \tau$ are in $M(\pcspsupport{X_1}, \pcspsupport{X_2})$,
  so $\tau \subseteq M(\pcspsupport{X_1}, \pcspsupport{X_2})$.
\end{proof}

\subsection*{Gradients}

Suppose we have a set of reference distributions $\{p_i\}$ with corresponding taxon sets $\{X_i\}$.
Suppose further that these distributions are parameterized by SBNs.
Our goal is to find an SBN $q$ on taxon set $X = \bigcup_i X_i$
that minimizes the sum of KL-divergences between each $p_i$
and the restriction of $q$ to $X_i$.
If we can calculate the gradient of $\sum_i \KLD{p_i}{\rest{q}{X_i}}$
with respect to the parameters of our SBN,
then we can use gradient descent to minimize our objective
and optimize our supertree distribution.
In this section we describe how to perform such
gradient calculation in the CCD and the SCD parameterizations.

\subsubsection*{CCD parameterizations}

Under CCDs, we parameterize the distribution of subsplits $s$ conditional on their clade $U(s)$
using a \emph{softmax} transformation of a parameter vector $\mathbf{v} = \{v_{s}\}$, i.e.
\[
q(s \mid U(s)) = \frac{\exp({v_{s}})}{\sum_{s':U(s')=U(s)}\exp({v_{s'}})}.
\]
We choose this softmax parameterization in order to have
clean derivatives with respect to our parameters
and facilitate taking the gradient of our objective function.
We use the shorthand
$\drv{s}f(\mathbf{v}) = \frac{\partial}{\partial v_{s}} f(\mathbf{v})$
for the derivative of a function with respect to one of our CCD parameters.

For subsplits $s$ and $s'$,
a standard derivative result for softmax parameters gives us,
\begin{equation} \label{eq:ccd_cond_deriv}
  \drv{s'} q(s \mid U(s)) = q(s' \mid U(s')) \left[ 1_{\{s=s'\}} - q(s \mid U(s')) \right],
\end{equation}
if $U(s) = U(s')$ and zero otherwise.
If we have CCD-parameterized SBNs $p$ and $q$ on tip sets $\bar{X}$ and $X$,
we can use Equation~\ref{eq:KL_CCD} to take the derivative of
the KL divergence between $p$ and $q$ restricted to $\bar{X}$,
\begin{align*}
  \drv{s'} \KLD{p}{\rest{q}{\bar{X}}}
  &= -\sum_{\bar{s}} p(\bar{s}) \: \drv{s'} \log \rest{q}{\bar{X}}(\bar{s}|U(\bar{s})) \\
  &= -\sum_{\bar{s}} \frac{p(\bar{s})}{\rest{q}{\bar{X}}(\bar{s}|U(\bar{s}))} \: \drv{s'} \, \rest{q}{\bar{X}}(\bar{s}|U(\bar{s})). \stepcounter{equation}\tag{\theequation}\label{eq:kld_deriv_ccd}
\end{align*}
We can then use Equations~\ref{eq:CCDUnconditional},
\ref{eq:CCDCladeUnconditional}, and~\ref{eq:CCDConditional}
to break the formula down to solely depend on
the derivative of the unconditional subsplit probability.
For subsplit $\bar{s}$ and clade $\bar{U}$ in $\bar{X}$,
\begin{align*}
  \drv{s'} \, \rest{q}{\bar{X}}(\bar{s}|U(\bar{s}))
  &= \frac{ \rest{q}{\bar{X}}(U(\bar{s})) \cdot \drv{s'} \rest{q}{\bar{X}}(\bar{s})
          - \rest{q}{\bar{X}}(\bar{s}) \cdot \drv{s'} \rest{q}{\bar{X}}(U(\bar{s})) }
          { \rest{q}{\bar{X}}(U(\bar{s}))^{2} } \\
  &= \frac{ \drv{s'} \rest{q}{\bar{X}}(\bar{s})
     - \rest{q}{\bar{X}}(\bar{s}|U(\bar{s})) \cdot
     \drv{s'} \rest{q}{\bar{X}}(U(\bar{s})) }
     { \rest{q}{\bar{X}}(U(\bar{s})) }, \stepcounter{equation}\tag{\theequation}\label{eq:ccd_rest_cond_deriv}\\
  \drv{s'} \rest{q}{\bar{X}}(\bar{s})
  &= \sum_{\rest{s}{\bar{X}} = \bar{s}} \drv{s'} q(s), \\
  \drv{s'} \rest{q}{\bar{X}}(\bar{U})
  &= \sum_{\bar{s}:U(\bar{s})=\bar{U}} \drv{s'} \rest{q}{\bar{X}}(\bar{s})
   = \sum_{\bar{s}:U(\bar{s})=\bar{U}} \sum_{\rest{s}{\bar{X}} = \bar{s}} \drv{s'} q(s).
\end{align*}
Thus all of our derivatives depend on the derivative
of the unconditional subsplit probability.
If we are taking a derivative with respect to $v_{s'}$,
we can then use the law of total probability and SBN conditional independence
to split our unconditional subsplit probabilities
into the collection of paths that pass through $U(s')$ and
the paths that do not.
We use the notation $C_{s'} \coloneqq q(U(s') \notin \{\rootparent{X} \ancdec s\})$
to capture the paths that do not pass through $U(s')$
and will therefore be a constant with respect to $v_{s'}$.
\begin{align*}
  q(s) &= q(U(s') \ancdec s) + q(U(s') \notin \{\rootparent{X} \ancdec s\}) \\
  &= \sum_{s'':U(s') \pach s''} q(U(s') \pach s'' \ancdec s) + C_{s'} \\
  &= q(U(s')) \sum_{s'':U(s') \pach s''} q(s'' | U(s')) \, q(s'' \ancdec s | s'') + C_{s'}.
\end{align*}

Finally, we can use Equation \ref{eq:ccd_cond_deriv}
to express the unconditional probability derivative in terms of
conditional and unconditional probabilities,
readily available from the SBN itself.
\begin{align*}
  \drv{s'} q(s)
  &= q(U(s')) \sum_{s'':U(s') \pach s''} \drv{s'|U(s')} \, q(s'' | U(s')) \, q(s'' \ancdec s \mid s'') \\
  &= q(U(s')) \sum_{s'':U(s') \pach s''} q(s'|U(s')) \left[ 1_{\{s''=s'\}} - q(s''|U(s')) \right]  q(s'' \ancdec s \mid s'') \\
  &= q(U(s')) q(s'|U(s')) \left[q(s' \ancdec s \mid s') - \sum_{s'':U(s') \pach s''} q(s''|U(s')) q(s'' \ancdec s \mid s'') \right] \\
  &= q(U(s')) q(s'|U(s')) \left[q(s' \ancdec s \mid s') - q(U(s') \ancdec s \mid U(s')). \right] \stepcounter{equation}\tag{\theequation}\label{eq:deriv_subsplit_CCD}
\end{align*}
Combining Equations \ref{eq:kld_deriv_ccd}, \ref{eq:ccd_rest_cond_deriv}, and \ref{eq:CCDCladeUnconditional} gives us a relatively succinct formula for the derivative,
\begin{align*}
  \drv{s'} \KLD{p}{\rest{q}{\bar{X}}}
  &= -\sum_{\bar{s}} \frac{p(\bar{s})}{\rest{q}{\bar{X}}(\bar{s}|U(\bar{s}))} \: \drv{s'} \, \rest{q}{\bar{X}}(\bar{s}|U(\bar{s})) \\
  &= -\sum_{\bar{s}} \frac{ p(\bar{s}) } { \rest{q}{\bar{X}}(\bar{s}) }
  \left[ \drv{s'} \rest{q}{\bar{X}}(\bar{s}) - \rest{q}{\bar{X}}(\bar{s}|U(\bar{s})) \cdot \drv{s'} \rest{q}{\bar{X}}(U(\bar{s})) \right] \\
  &= -\sum_{\bar{s}} \frac{ p(\bar{s}) } { \rest{q}{\bar{X}}(\bar{s}) } \drv{s'} \rest{q}{\bar{X}}(\bar{s}) +
  \sum_{\bar{s}} \frac{ \drv{s'} \rest{q}{\bar{X}}(U(\bar{s})) } { \rest{q}{\bar{X}}(U(\bar{s})) } p(\bar{s}) \\
  &= \sum_{\bar{U}} \frac{ \drv{s'} \rest{q}{\bar{X}}(\bar{U}) } { \rest{q}{\bar{X}}(\bar{U}) } \sum_{\bar{s}:U(\bar{s})=\bar{U}} p(\bar{s})
   - \sum_{\bar{s}} \frac{ p(\bar{s}) } { \rest{q}{\bar{X}}(\bar{s}) } \drv{s'} \rest{q}{\bar{X}}(\bar{s}) \\
  &= \sum_{\bar{U}} \frac{ p(\bar{U}) } { \rest{q}{\bar{X}}(\bar{U}) } \drv{s'} \rest{q}{\bar{X}}(\bar{U})
   - \sum_{\bar{s}} \frac{ p(\bar{s}) } { \rest{q}{\bar{X}}(\bar{s}) } \drv{s'} \rest{q}{\bar{X}}(\bar{s}).
\end{align*}
This form displays one of the derivative's natural symmetries between clade and subsplit probabilities and derivatives.
However, if implemented naively, this form may result in iterating over the subsplit support multiple times.
We address this by exchanging summations,
\begin{align*}
  \drv{s'} \KLD{p}{\rest{q}{\bar{X}}}
  &= \sum_{\bar{U}} \frac{ p(\bar{U}) } { \rest{q}{\bar{X}}(\bar{U}) } \sum_s \drv{s'} q(s) 1_{\{U(\rest{s}{\bar{X}}) = \bar{U}\}} \\
   & \quad - \sum_{\bar{s}} \frac{ p(\bar{s}) } { \rest{q}{\bar{X}}(\bar{s}) } \sum_s \drv{s'} q(s) 1_{\{ \rest{s}{\bar{X}} = \bar{s} \}} \\
  &= \sum_s \left[ \frac{ p(U(\rest{s}{\bar{X}})) } { \rest{q}{\bar{X}}(U(\rest{s}{\bar{X}})) }
   - \frac{ p(\rest{s}{\bar{X}}) } { \rest{q}{\bar{X}}(\rest{s}{\bar{X}}) } \right] \drv{s'} q(s).
\end{align*}
This form clearly shows the algorithmic complexity of the gradient computation
as $O(\nss^2)$ where $\nss$ is the number of subsplits in the support,
since both the summation and the derivative traverse every subsplit.

\subsubsection*{SCD parameterizations}

Under SCDs, we parameterize the distribution of child subsplits $s$
conditional on their parent subsplit and clade $\sscl{t}{U(s)}$
with parameter vector $\mathbf{v} = \{v_{s|t}\}$, i.e.
\[
q(s|\sscl{t}{U(s)}) = \frac{\exp({v_{s|t}})}{\sum_{s':\sscl{t}{U(s)}\pach s'}\exp({v_{s'|t}})}.
\]
We use the shorthand
$\drv{s|t}f(\mathbf{v}) = \frac{\partial}{\partial v_{s|t}} f(\mathbf{v})$
for the derivative of a function with respect to one of our SCD parameters.
For subsplits $s$, $s'$, $t$, and $t'$,
our softmax derivative result is then,
\begin{equation*}
  \drv{s'|t'} q(s|t) = q(s'|\sscl{t'}{U(s)}) \left[ 1_{\{s=s'\}} - q(s|\sscl{t'}{U(s)}) \right],
  \stepcounter{equation}\tag{\theequation}\label{eq:scd_cond_deriv}
\end{equation*}
if $t = t'$ and zero otherwise.
The derivative of the KL divergence is,
\begin{align*}
  \drv{s'|t'} \KLD{p}{\rest{q}{\bar{X}}}
  &= -\sum_{(\bar{t} \pach \bar{s})} p(\bar{t} \pach \bar{s}) \: \drv{s'|t'} \log \rest{q}{\bar{X}}(\bar{s}|\bar{t}) \\
  &= -\sum_{(\bar{t} \pach \bar{s})} \frac{p(\bar{t} \pach \bar{s})}{\rest{q}{\bar{X}}(\bar{s}|\bar{t})} \: \drv{s'|t'} \, \rest{q}{\bar{X}}(\bar{s}|\bar{t}).
  \stepcounter{equation}\tag{\theequation}\label{eq:kld_deriv_scd}
\end{align*}
For PCSP $(\bar{t} \pach \bar{s})$ we see,
\begin{align*}
  \drv{s'|t'} \, \rest{q}{\bar{X}}(\bar{s}|\bar{t})
  &= \frac{ \rest{q}{\bar{X}}(\bar{t}) \cdot \drv{s'|t'} \rest{q}{\bar{X}}(\bar{t} \pach \bar{s})
          - \rest{q}{\bar{X}}(\bar{t} \pach \bar{s}) \cdot \drv{s'|t'} \rest{q}{\bar{X}}(\bar{t}) }
          { \rest{q}{\bar{X}}(\bar{t})^{2} } \\
  &= \frac{ \drv{s'|t'} \rest{q}{\bar{X}}(\bar{t} \pach \bar{s})
     - \rest{q}{\bar{X}}(\bar{s}|\bar{t}) \cdot \drv{s'|t'} \rest{q}{\bar{X}}(\bar{t}) }
     { \rest{q}{\bar{X}}(\bar{t}) },
     \stepcounter{equation}\tag{\theequation}\label{eq:scd_rest_cond_deriv}\\
  \drv{s'|t'} \rest{q}{\bar{X}}(\bar{t})
  &= \sum_{\rest{a}{\bar{X}} = \bar{t}} \drv{s'|t'} q(a), \\
  \drv{s'|t'} \rest{q}{\bar{X}}(\bar{t} \pach \bar{s})
  &= \sum_{\rest{a}{\bar{X}} = \bar{t}} \sum_{\rest{d}{\bar{X}} = \bar{s}} \drv{s'|t'} q(a \ancdec d).
\end{align*}
One building block we need is
\begin{align*}
  \derv{q}{s'|t'}{a}
  &\coloneqq \drv{s'|t'} q(\sscl{t'}{U(s')} \ancdec a \mid \sscl{t'}{U(s')}) \\
  &= \sum_{s'':\sscl{t'}{U(s')} \pach s''} \drv{s'|t'} q(s'' | \sscl{t'}{U(s')}) q(s'' \ancdec a \mid s'') \\
  &= \sum_{s'':\sscl{t'}{U(s')} \pach s''} q(s'|\sscl{t'}{U(s')}) \left[ 1_{\{s''=s'\}} - q(s''|\sscl{t'}{U(s')}) \right]  q(s'' \ancdec a \mid s'') \\
  &= q(s'|\sscl{t'}{U(s')}) \left[q(s' \ancdec a \mid s') - \sum_{s'':\sscl{t'}{U(s')} \pach s''} q(s''|\sscl{t'}{U(s')}) q(s'' \ancdec a \mid s'') \right] \\
  &= q(s'|\sscl{t'}{U(s')}) \left[q(s' \ancdec a \mid s') - q(\sscl{t'}{U(s')} \ancdec a \mid \sscl{t'}{U(s')}) \right].
\end{align*}
Note that computing $\derv{q}{s'|t'}{a}$ is a constant time calculation
after accumulating a table of path probabilities in linear time before the gradient calculation.

Following an argument similar to Equation~\ref{eq:deriv_subsplit_CCD},
we drop terms that are constant with respect to $s'|t'$ and see that,
\begin{align*}
  \drv{s'|t'} q(a)
  &= q(t') \drv{s'|t'} q(\sscl{t'}{U(s')} \ancdec a \mid \sscl{t'}{U(s')}) \\
  &= q(t') \derv{q}{s'|t'}{a}. \stepcounter{equation}\tag{\theequation}\label{eq:deriv_subsplit_SCD}
\end{align*}
Furthermore, by identical reasoning we calculate the derivative of the path probabilities,
\begin{align*}
  \drv{s'|t'} q(a \ancdec d)
  &= \drv{s'|t'} \left[ q(a) q(a \ancdec d | a) \right] \\
  &= q(a) \drv{s'|t'}  q(a \ancdec d | a) + \drv{s'|t'} q(a) q(a \ancdec d | a) \\
  &= q(a) q(a \ancdec t' | a) \drv{s'|t'} q(\sscl{t'}{U(s')} \ancdec d \mid \sscl{t'}{U(s')}) \\
  & \qquad + q(t') \derv{q}{s'|t'}{a} q(a \ancdec d | a) \\
  &= q(a) q(a \ancdec t' | a) \derv{q}{s'|t'}{d} \\
  & \qquad + q(t') \derv{q}{s'|t'}{a} q(a \ancdec d | a).
   \stepcounter{equation}\tag{\theequation}\label{eq:deriv_path_SCD}
\end{align*}

We combine Equations~\ref{eq:kld_deriv_scd} and~\ref{eq:scd_rest_cond_deriv} to find our KL derivative,
\begin{align*}
  \drv{s'|t'} \KLD{p}{\rest{q}{\bar{X}}}
  &= -\sum_{(\bar{t} \pach \bar{s})}
  \frac{p(\bar{t} \pach \bar{s})}{\rest{q}{\bar{X}}(\bar{s}|\bar{t})} \:
  \drv{s'|t'} \, \rest{q}{\bar{X}}(\bar{s}|\bar{t}), \\
  &= -\sum_{(\bar{t} \pach \bar{s})}
  \frac{p(\bar{t} \pach \bar{s})}{\rest{q}{\bar{X}}(\bar{s}|\bar{t})} \:
  \frac{1}{\rest{q}{\bar{X}}(\bar{t})}
  \left[ \drv{s'|t'} \rest{q}{\bar{X}}(\bar{t} \pach \bar{s})
     - \rest{q}{\bar{X}}(\bar{s}|\bar{t}) \cdot \drv{s'|t'} \rest{q}{\bar{X}}(\bar{t}) \right] \\
  &= \sum_{(\bar{t} \pach \bar{s})}
  \frac{p(\bar{t})}{\rest{q}{\bar{X}}(\bar{t})}
  \frac{p(\bar{s}|\bar{t})}{\rest{q}{\bar{X}}(\bar{s}|\bar{t})}
  \left[ \rest{q}{\bar{X}}(\bar{s}|\bar{t}) \cdot \drv{s'|t'} \rest{q}{\bar{X}}(\bar{t})
  - \drv{s'|t'} \rest{q}{\bar{X}}(\bar{t} \pach \bar{s}) \right] \\
  &= \sum_{\bar{t}}
  \frac{p(\bar{t})}{\rest{q}{\bar{X}}(\bar{t})}
  \drv{s'|t'} \rest{q}{\bar{X}}(\bar{t})
  \sum_{\bar{s}:\bar{t} \pach \bar{s}} p(\bar{s}|\bar{t}) \\
  & \quad -\sum_{(\bar{t} \pach \bar{s})}
  \frac{p(\bar{t})}{\rest{q}{\bar{X}}(\bar{t})}
  \frac{p(\bar{s}|\bar{t})}{\rest{q}{\bar{X}}(\bar{s}|\bar{t})}
  \drv{s'|t'} \rest{q}{\bar{X}}(\bar{t} \pach \bar{s}) \\
  &= \sum_{\bar{t}} \children{\bar{t}}
  \frac{p(\bar{t})}{\rest{q}{\bar{X}}(\bar{t})}
  \drv{s'|t'} \rest{q}{\bar{X}}(\bar{t}) \\
  & \quad -\sum_{(\bar{t} \pach \bar{s})}
  \frac{p(\bar{t})}{\rest{q}{\bar{X}}(\bar{t})}
  \frac{p(\bar{s}|\bar{t})}{\rest{q}{\bar{X}}(\bar{s}|\bar{t})}
  \drv{s'|t'} \rest{q}{\bar{X}}(\bar{t} \pach \bar{s}) \\
  &= \sum_{\bar{t}} \children{\bar{t}}
  \frac{p(\bar{t})}{\rest{q}{\bar{X}}(\bar{t})}
  \sum_{\rest{a}{\bar{X}} = \bar{t}} \drv{s'|t'} q(a) \\
  & \quad -\sum_{(\bar{t} \pach \bar{s})}
  \frac{p(\bar{t})}{\rest{q}{\bar{X}}(\bar{t})}
  \frac{p(\bar{s}|\bar{t})}{\rest{q}{\bar{X}}(\bar{s}|\bar{t})}
  \sum_{\rest{a}{\bar{X}} = \bar{t}} \sum_{\rest{d}{\bar{X}} = \bar{s}}
  \drv{s'|t'} q(a \ancdec d),
\end{align*}
where $\children{\bar{t}}$ is the number of child clades of $\bar{t}$ of size 2 or larger,
and therefore have a probability distribution of child subsplits to sum over.
After the linear pass through the support accumulating path probabilities,
the algorithmic efficiency of this calculation is
$O(\npcsp \cdot \npcsprest{\bar{X}})$,
where $\npcsp$ is the number of PCSPs in the support,
and $\npcsprest{\bar{X}}$ is the number of paths in the support
that restrict to a PCSP on tip set $\bar{X}$.


\section*{Results}


\subsection*{Simulated Data}

We begin exploring the effectiveness of \texttt{vbsupertree}
through a simulation study.
We sample a phylogenetic tree with 40 tips from the
classical isochronous, constant effective population size coalescent,
and simulate a sequence data alignment using the Jukes-Cantor 1969 model
\citep{jukes1969evolution}.
We choose a tip to remove from the alignment for our first reference dataset,
and repeat the process on a different tip for our second reference.
We approximate the true posterior and our two references
by running Markov chain Monte Carlo \citep{hastings1970monte}
using the phylogenetic software BEAST \citep{drummond2007beast,Suchard2018-eo}
on our three sequence datasets.
We run BEAST for $10^7$ steps,
remove 50\% burn-in from the beginning,
and subsample every 1000th tree to reduce autocorrelation,
resulting in a 5000 tree posterior samples.
Since we are not guaranteed to see every credible tree topology in every run due to the size of tree space,
we trim out all tree topologies that only appear in a given BEAST output once,
in order to increase the proportion of PCSPs in common
(under the appropriate restriction) between the three posterior samples.
We train rooted, SCD-parameterized SBNs
to use as our ground truth distribution and two reference distributions.
Finally, we trim our SBNs in order to make the supports compatible for KL-divergence calculation.
We trim any PCSPs in our ground truth and our references
that are not covered by the appropriate restriction of our mutualized support
(no restriction in the case of the ground truth).

Applying \texttt{vbsupertree} to the two generated reference distribution
leads to quick convergence of the loss function to a small value,
as seen in Figure~\ref{fig:vbsupertree_sim}, left panel.
Additionally, knowing the true posterior, we chart
the progression of the KL-divergence of our supertree SBN versus the truth,
resulting in the right panel.

\begin{figure}[htbp]
  \centering
  \includegraphics[width=0.95\textwidth]{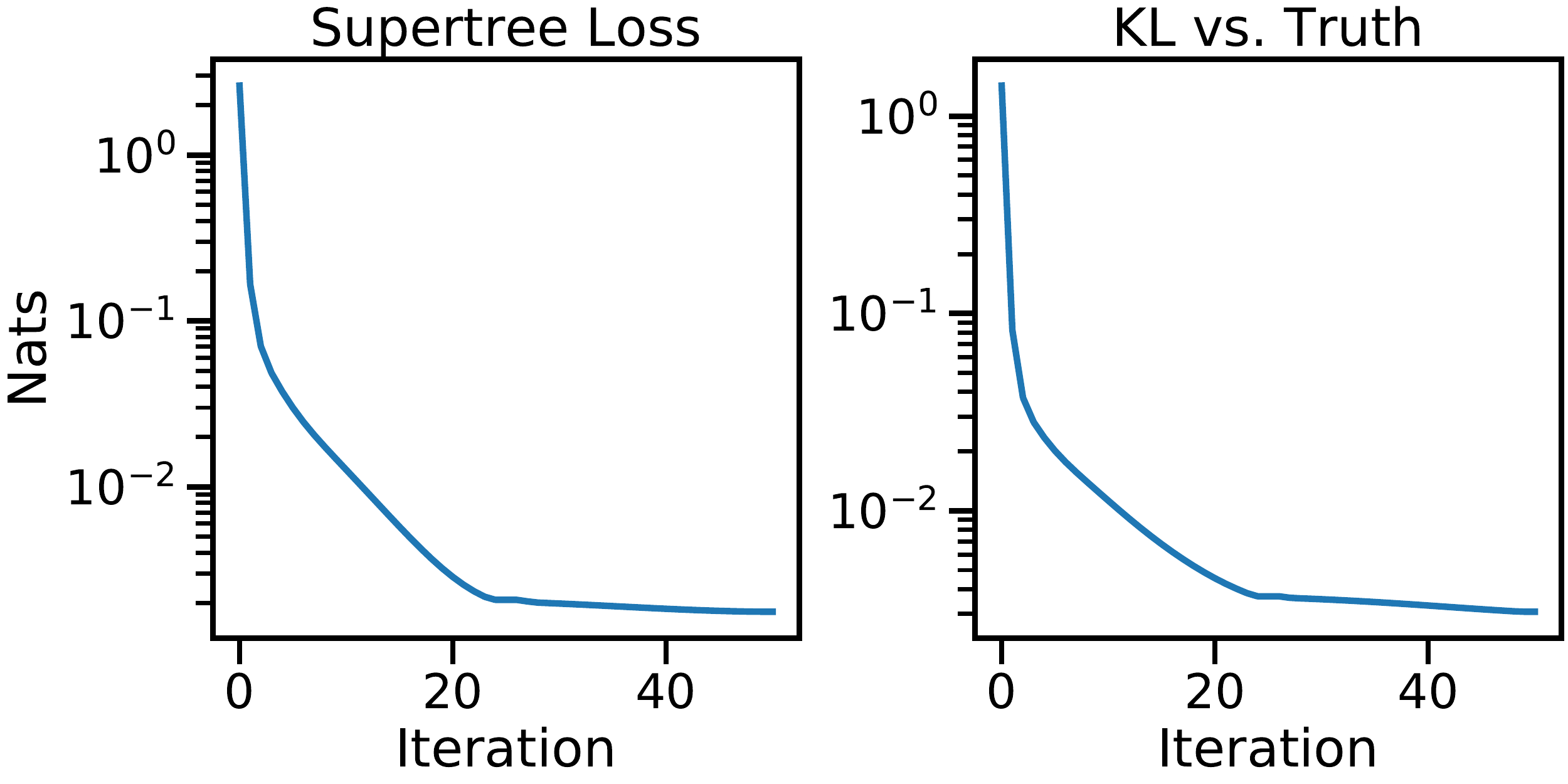}
  \caption{The left panel shows the progression of the loss function over 50 iterations of \texttt{vbsupertree}.
  The right panel shows the progression of the KL-divergence versus the truth.}
  \label{fig:vbsupertree_sim}
\end{figure}

\subsection*{Real World Data}

For an analysis on real world data,
we select 30 well-differentiated hepatitis C virus (HCV) sequences
from the alignment previously analyzed by \citet{pybus2003epidemiology} and others.
We choose a tip to remove from the alignment for our first reference dataset,
and repeat the process on a different tip for our second reference.
From this stage forward, our approach is identical to our simulation study.
We approximate the true posterior and our two references
by running BEAST on our three sequence datasets.
We run BEAST for $10^9$ steps,
remove 50\% burn-in from the beginning,
and subsample every 5000th tree to reduce autocorrelation,
resulting in a 100000 tree posterior samples.
We trim out all tree topologies that only appear in a given BEAST output once,
We train rooted, SCD-parameterized SBNs
to use as our ground truth distribution and two reference distributions.
Finally, we trim our SBNs in order to make the supports compatible for KL-divergence calculation.
We trim any PCSPs in our ground truth and our references
that are not covered by the appropriate restriction of our mutualized support.

Applying \texttt{vbsupertree} to the two generated reference distribution
leads to quick convergence of the loss function to a small value,
as seen in Figure~\ref{fig:vbsupertree_hcv}, left panel.
Additionally, knowing the true posterior, we chart
the progression of the KL-divergence of our supertree SBN versus the truth,
resulting in the right panel.

\begin{figure}[htbp]
  \centering
  \includegraphics[width=0.95\textwidth]{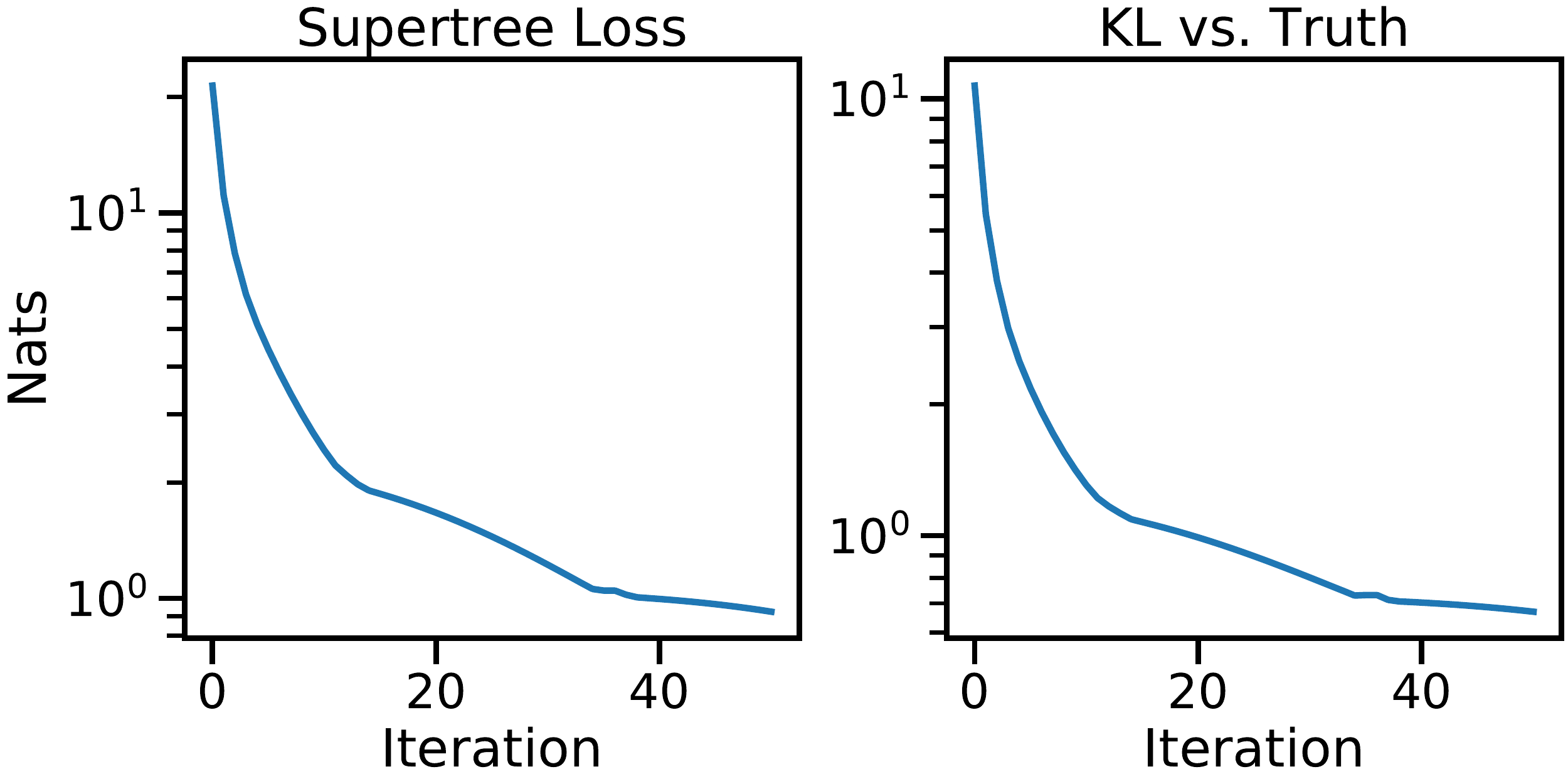}
  \caption{The left panel shows the progression of the loss function over 50 iterations of \texttt{vbsupertree}.
  The right panel shows the progression of the KL-divergence versus the truth.}
  \label{fig:vbsupertree_hcv}
\end{figure}

\section*{Discussion}

In this paper, we lay out an SBN-based framework for
generating supertree supports and training variational supertree distributions.
We apply our method to simulated sequence data and find that
it trains an SBN that very closely approximates our target posterior distribution.
We also apply our method to a subset of a well-known HCV dataset,
and successfully train it to approximate our ground truth distribution.

Although the work of \citet{Ronquist2004-rr} described in the introduction is the closest work to that presented here, two other lines of research deserve mention in this context.
First, \citet{De_Oliveira_Martins2016-mi} derive a Bayesian extension of previous work on maximum\hyp{}likelihood supertrees \citep{Steel2008-pn}.
In this strategy, one posits a likelihood model based on measures of disagreement between trees, such as an exponential likelihood in terms of some distance between tree topologies.
This method is interesting in that it can incorporate a number of distances representing various aspects of tree disagreement \citep{De_Oliveira_Martins2016-mi}, however, this is a different than the direct goal of reconstructing a posterior distribution on taxon set $X$ given its projections onto subsets as we describe below.
Our objective directly phrases a goal appropriate for divide-and-conquer Bayesian phylogenetics.

Another related line of research concerns sequential Monte Carlo inference by subtree merging \citep{Bouchard-Ct2012-zu,Wang2015-nm}.
The state of such a sampler is described by a population of ``particles,'' each of which consists of a collection of rooted trees on disjoint taxon subsets such that the union of the tree tips is the entire taxon set.
In each step of the algorithm, particles are chosen from the previous generation, and for each particle a pair of subtrees are merged.
These probabilistic choices and mergings are designed carefully such that after completion of all of the steps one obtains a sample from the phylogenetic posterior distribution.
This method is in a sense a type of divide-and-conquer algorithm in that it finds solutions to phylogenetic problems on subsets of taxa before finding the entire posterior.
However, it differs significantly from our current goal in that we assume that the taxon subsets and the posterior distributions on subtrees are delivered as part of the problem statement, whereas phylogenetic SMC ingests raw molecular sequence data.

One common obstacle for supertree methods is the fact that
the compatibility of $k$ tree topologies on $k$ tip sets
cannot be checked in polynomial time in $k$ \citep{Steel1992-by}.
For the methods we present, holding ourselves to two tip sets at a time,
this is not an issue.
Our subsplit- and PCSP-based approaches pool our topologies
into two sets which effectively sets $k=2$.
It is for this reason we propose using a one-at-a-time approach
for using our supertree support mutualization methods on $k>2$ tip sets.
We anticipate exploring the properties of one-at-a-time versus all-at-once
mutualization in future work.

One caveat for our supertree support mutualization methods
arises when the reference supports do not completely cover
the true restricted supports.
When the references cover the truth,
our results guarantee that the mutualized support contains
every topology that we require without containing any extraneous elements.
However, if the reference supports are missing elements from the true supports,
then topologies will go missing from the mutualized supertree support
and it is not guaranteed to cover the true support.
Unfortunately, most tree-based Bayesian analyses will have
enormous posterior topology supports,
and Monte Carlo based methods will collect only a sample from the larger posterior.
Thus in future work, we will loosen our inclusion criteria for our supertree support methods while still attempting to keep the mutual support as small as possible.

In general we view this work as providing a foundation for divide-and-conquer Bayesian phylogenetics.
To make this a complete method, we will also require methods to merge variational branch length distributions~\citep{vbpi}.
Further refinement of these merged distributions with the complete data set, in terms of both support and continuous parameters, will likely be required.

\section*{Acknowledgements}
The authors thank Alexei Drummond for emphasizing the importance of this problem, and Mike Steel for a discussion of problem complexity.

We also thank the larger community of researchers engaged in variational Bayes phylogenetic inference, including
Mathieu Fourment,
Xiang Ji,
Seong-Hwan Jun,
Ognian Milanov,
Hassan Nasif,
Christiaan Swanepoel,
and
Marc Suchard.

This work supported by NSF grants CISE-1561334 and CISE-1564137 as well as NIH U54 grant GM111274.
The research of Frederick Matsen was supported in part by a Faculty Scholar grant from the Howard Hughes Medical Institute and the Simons Foundation.

\bibliographystyle{plainnat}
\bibliography{main}

\appendix

\section*{Appendix}

\subsection*{Mutual PCSP support example}

We will now illustrate Algorithm~\ref{alg:mutual_pcsp} with a simple yet nontrivial example.
To illustrate specific clades, subsplits, and PCSPs, we introduce some shorthand notations for this section.
Tree tips will be represented by capital letters, such as $A$ and $B$.
Clades will be represented by concatenated tips, such as $ABD$ and $ACD$.
Subsplits will be represented by two clades separated by a colon, such as $AB:CD$.
Subsplits focusing on a specific child clade will be represented by two clades separated by a slash, focusing on the latter clade, such as $CD/AB$.
PCSPs will then look like $CD/AB \pach A:B$.

\begin{figure}[htbp]
  \centering
  \includegraphics[width=0.9\textwidth]{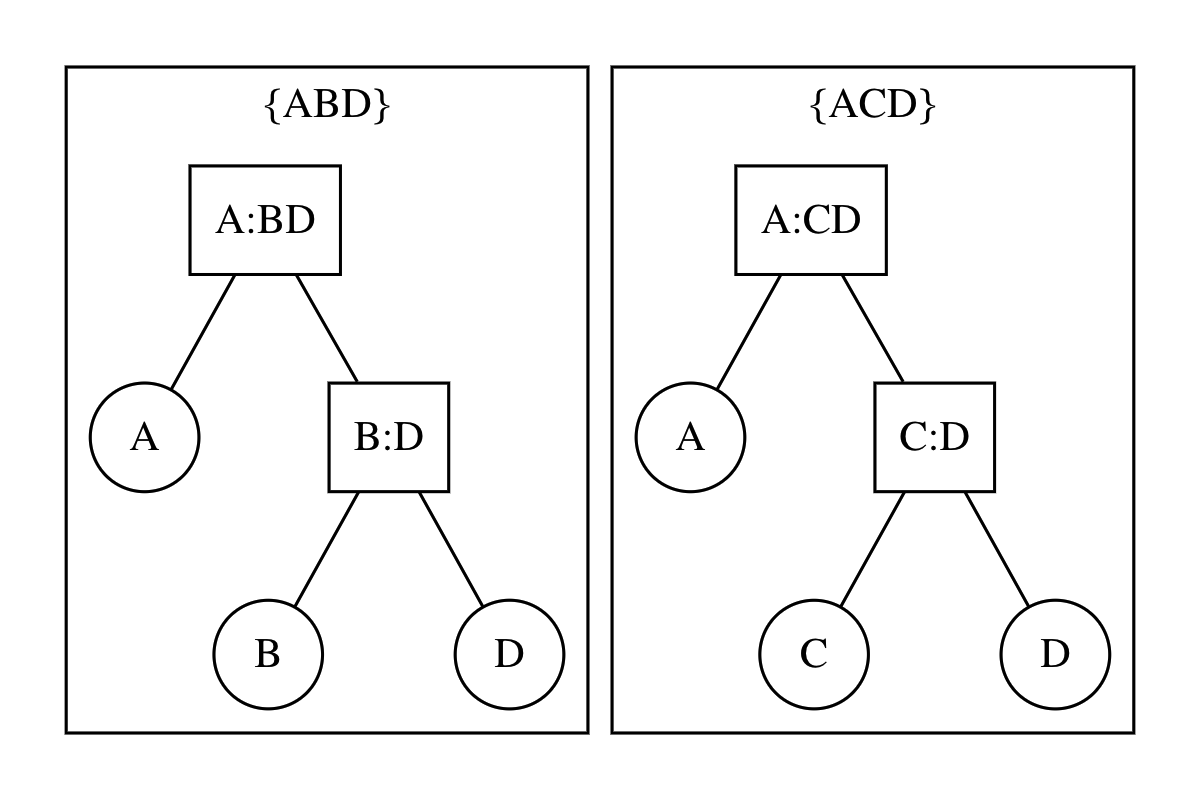}
  \caption{Example starting trees for finding the mutual PCSP support.
  }%
  \label{fig:mutual_pcsp_trees}
\end{figure}

Suppose we are given the trees depicted in Figure~\ref{fig:mutual_pcsp_trees}.
The tree with tips $ABD$ results in the reference PCSP support
\begin{align*}
  \pcspsupport{ABD} = \{/ABD &\pach A:BD,\\
    A/BD &\pach B:D\},
\end{align*}
and the tree with tips $ACD$ results in
\begin{align*}
  \pcspsupport{ACD} = \{/ACD &\pach A:CD,\\
    A/CD &\pach C:D\},
\end{align*}
omitting singletons and the empty subsplit for brevity.

The union of the tip sets is $ABCD$.
The mutual PCSP support algorithm begins with stack $[(/ABCD, /ABD, /ACD)]$.

We consider the children (including the trivial child) of $/ABD$ which are
$\{A:BD, \colon ABD\}$ and the children of $/ACD$, namely $\{A:CD, \colon ACD\}$.
The only nontrivial subsplit this results in via $\subspcross$ is $A:BCD$.

We add $/ABCD \pach A:BCD$ to the output.

We see $A:BCD$ restricts to $A:BD$ and $A:CD$,
both of which are subsplits in their respective reference supports.
The only child clade of $A:BCD$ that is size 2 or larger is $BCD$,
so we push $(A/BCD, A/BD, A/CD)$ to the stack.

We pop $(A/BCD, A/BD, A/CD)$.

We combine the subsplits $\{B:D, \colon BD\}$ and $\{C:D, \colon CD\}$ via $\subspcross$ as before,
resulting in a set of potential children: $\{BC:D, B:CD, BD:C\}$.

We add $A/BCD \pach BC:D$ to the output.

We see $BC:D$ restricts to $B:D$ and $C:D$,
so we push $(D/BC, D/B, D/C)$ to the stack.

We add $A/BCD \pach B:CD$ to the output.

We see $BC:D$ restricts to $B:D$ and $\colon CD$.
However, $\colon CD$ is not in $\pcspsupport{ACD}$,
triggering the ``else'' clause,
so we push $(B/CD, B/D, A/CD)$ to the stack.

We add $A/BCD \pach BD:C$ to the output.

We see $BD:C$ restricts to $:BD$ (not in $\pcspsupport{ABD}$) and $C:D$,
so we push $(C/BD, A/BD, C/D)$ to the stack.

We pop $(C/BD, A/BD, C/D)$.

We combine the subsplits $\{B:D, \colon BD\}$ and $\{\colon D\}$,
resulting in $B:D$.

We add $C/BD \pach B:D$ to the output.
No child clades are size 2 or larger, so we do not push anything to the stack.

We pop $(B/CD, B/D, A/CD)$.

We combine the subsplits $\{\colon D\}$ and $\{C:D, \colon CD\}$,
resulting in $C:D$.

We add $B/CD \pach C:D$ to the output.

We pop $(D/BC, D/B, D/C)$.

We combine the subsplits $\{\colon B\}$ and $\{\colon C\}$,
resulting in $B:C$.

We add $D/BC \pach B:C$ to the output.

The stack is empty, so the algorithm terminates here.
The final output is the PCSP support
\begin{align*}
 \{&/ABCD \pach A:BCD, \\
   &A/BCD \pach BC:D, \\
   &A/BCD \pach B:CD, \\
   &A/BCD \pach BD:C, \\
   &C/BD \pach B:D, \\
   &B/CD \pach C:D, \\
   &D/BC \pach B:C\}.
\end{align*}

\begin{figure}[htbp]
  \centering
  \includegraphics[width=0.9\textwidth]{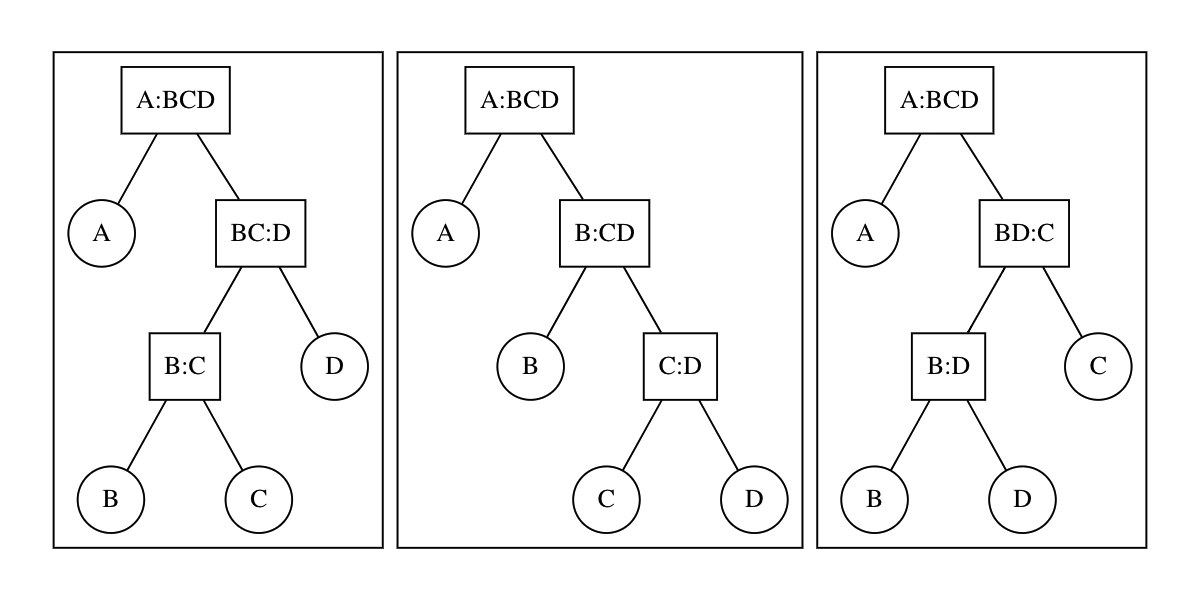}
  \caption{Trees resulting from the mutual PCSP algorithm example applied to the trees from Figure~\ref{fig:mutual_pcsp_trees}.
  }%
  \label{fig:mutual_pcsp_results}
\end{figure}

This PCSP support results in the trees depicted in Figure~\ref{fig:mutual_pcsp_results}.
Note that each tree restricts to the appropriate reference trees
in Figure~\ref{fig:mutual_pcsp_trees}, as required.

\end{document}